\pdfoutput=1
\documentclass[final,11pt]{article}

\usepackage[thmstyles]{andrews}
\usepackage{andrews_bib}

\usepackage[inline]{showlabels}

\usepackage[style=alphabetic,
backend=biber,
backref,
url=false,
doi=true,
isbn=false,
useprefix=true,
eprint=false,
maxbibnames=99,
maxcitenames=5,
mincitenames=3,
maxalphanames=5,
minalphanames=3]{biblatex}
\renewcommand*{\bibfont}{\small}

\newcommand{\Jac}[1]{\mathrm{Jac(}#1\mathrm{)}}
\newcommand{\Ann}[1]{\mathrm{Ann(}#1\mathrm{)}}
\newcommand{\V}{\mathbf{V}}
\newcommand{\I}{\mathbf{I}}

\DeclareMathOperator{\res}{res}
\DeclareMathOperator{\size}{size}

\title{Algebraic Pseudorandomness in $\VNC^0$}
\author{Robert Andrews\thanks{Cheriton School of Computer Science, University of Waterloo. Part of this work was done at the Institute for Advanced Study and was supported by NSF grant CCF-1900460 and the Erik Ellentuck Endowed Fellowship Fund. Email: randrews@uwaterloo.ca.}}
\date{May 15, 2025}

\begin{document}

\maketitle

\begin{abstract}
    We study the arithmetic complexity of hitting set generators, which are pseudorandom objects used for derandomization of the polynomial identity testing problem.
    We give new explicit constructions of hitting set generators whose outputs are computable in $\VNC^0$, i.e., can be computed by arithmetic formulas of constant size.
    Unconditionally, we construct a $\VNC^0$-computable generator that hits arithmetic circuits of constant depth and polynomial size.
    We also give conditional constructions, under strong but plausible hardness assumptions, of $\VNC^0$-computable generators that hit arithmetic formulas and arithmetic branching programs of polynomial size, respectively.
    As a corollary of our constructions, we derive lower bounds for subsystems of the Geometric Ideal Proof System of Grochow and Pitassi.

    Constructions of such generators are implicit in prior work of Kayal on lower bounds for the degree of annihilating polynomials. 
    Our main contribution is a construction whose correctness relies on circuit complexity lower bounds rather than degree lower bounds.
\end{abstract}

\section{Introduction}

\subsection{Polynomial Identity Testing}

Algebraic complexity is a vibrant subarea of complexity theory that studies computation of polynomial and rational functions using basic arithmetic operations.
Like boolean complexity theory, algebraic complexity enjoys a rich theory of pseudorandomness, with the polynomial identity testing (PIT) problem playing a central role.
The input to the PIT problem is an arithmetic circuit, and the goal is to decide whether the circuit computes the identically zero polynomial.
This problem can be efficiently solved with randomness using the Schwartz--Zippel lemma \cite{Schwartz80,Zippel79}, which says that if a degree-$d$ polynomial $f$ is nonzero, then with probability at least $1/2$ a randomly-chosen point from a grid of side-length $2d$ will lead to a nonzero evaluation of $f$.
A great deal of work has gone into designing efficient deterministic algorithms for polynomial identity testing, leading to beautiful constructions and connections to other areas of computer science and mathematics.

Algorithms for PIT are often designed by constructing \emph{hitting set generators}, which play a role analogous to pseudorandom generators in boolean complexity.
For a complexity class $\mathscr{C}$, a hitting set generator $\mathcal{G}$ for $\mathscr{C}$ is a (family of) polynomial map(s) $\mathcal{G} : \F^{\ell} \to \F^n$ such that for every polynomial $f \in \mathscr{C}$, we have $f = 0$ if and only if $f \circ \mathcal{G} = 0$.\footnote{Formally, we consider \emph{families} of polynomials $(f_1, f_2, \ldots)$ and families of polynomial maps $(\mathcal{G}_1, \mathcal{G}_2, \ldots)$, rather than a single polynomial $f$ and a single map $\mathcal{G}$. We will gloss over this distinction throughout the introduction for the sake of readability, focusing instead on single elements $f_n$ and $\mathcal{G}_n$ of the respective families.}
Conceptually, one can think of a generator as mapping $\ell$ truly random field elements to $n$ pseudorandom field elements.
If the generator $\mathcal{G}$ can be implemented by an efficient algorithm, then we are led to an improved deterministic PIT algorithm for $\mathscr{C}$-circuits: given a circuit $C$, test the composition $C \circ \mathcal{G}$ by brute force.
The running time of the na\"{i}ve brute-force algorithm for PIT has an exponential dependence on the number of variables, and the generator $\mathcal{G}$ improves this exponent from $n$ to the smaller $\ell$, a parameter referred to as the \emph{seed length} of the generator.
The usual aim is to construct generators with seed length that is as small as possible, since this is the most important parameter for improving the running time of the deterministic algorithm.

Numerous constructions of hitting set generators are known, both conditional and unconditional.
In this work, we will be interested in designing hitting set generators for strong circuit classes, at or beyond the frontier of our ability to prove super-polynomial lower bounds for explicit polynomials.\footnote{For results on polynomial identity testing for weaker circuit classes, we refer the reader to the recent survey of \textcite{DG24}, as well as the surveys of \textcite{Saxena09,Saxena14}.}
Starting with \textcite{KI04}, who adapted the Nisan--Wigderson \cite{NW94} generator to the algebraic setting, there has been a successful line of work constructing hitting set generators for strong circuit classes---including general, unrestricted circuits---under hardness assumptions \cite{DSY09,CKS19a,Andrews20,ST21c,GKSS22}.
In fact, some of the lower bounds assumed by these works have since been proven unconditionally, leading to new deterministic algorithms for PIT!
Specifically, the super-polynomial lower bounds of \textcite{LST21} against constant-depth arithmetic circuits, combined with the hardness-to-randomness theorem of \textcite{CKS19a}, led to the first deterministic sub-exponential time algorithm to test polynomial identities written as constant-depth circuits.

Recently, \textcite{CT25} raised a very interesting question about the complexity of hitting set generators.
They asked if it is possible to construct a hitting set generator that is computable in some small complexity class $\mathscr{C}$, yet appears pseudorandom to a larger complexity class $\mathscr{D} \supsetneq \mathscr{C}$.
This sort of parameter regime is common throughout cryptography, so they termed such a generator a \emph{cryptographic} hitting set generator.
The question of constructing cryptographic hitting set generators is an extremely interesting one, and is the focus of our work.
Aside from this question's inherent interest, one might expect the techniques underlying the construction of a cryptographic hitting set generator to have other applications within algebraic complexity theory.
For example, the analogous question of constructing low-complexity pseudorandom generators was answered in a beautiful work of \textcite{AIK06}, and the techniques therein have found numerous applications within cryptography and complexity, such as to the recent study of the range avoidance problem \cite{RSW22}.

Most known hardness-randomness connections in algebraic complexity cannot hope to produce a hitting set generator that operates in this cryptographic regime of parameters.
Almost all generators that come from algebraic hardness-randomness are \emph{reconstructive}, meaning that their correctness proofs follow a common template.
In a reconstructive proof of correctness, the generator $\mathcal{G}$ is proven correct by an argument of the following form: suppose some explicit polynomial $f$, such as the $n \times n$ permanent, is hard to compute.
Given a nonzero circuit $C$ of size $s$ that satisfies $C \circ \mathcal{G} = 0$, there is a reconstruction procedure that modifies $C$ into a circuit $C'$ of size $s + t$ that computes the supposedly-hard polynomial $f$.
If $s + t \ll \size(f)$, where $\size(f)$ is the complexity of $f$, then we arrive at a contradiction.
The upshot of this argument is that if the circuit $C$ is sufficiently small---in particular, when the circuit $C$ is of size $s \ll \size(f)$---the hardness of $f$ implies that the composition $C \circ \mathcal{G}$ must be nonzero, i.e., that $\mathcal{G}$ hits small circuits.
The bottleneck in this argument is that the generator $\mathcal{G}$ is usually defined using the hard polynomial $f$ (for example, $\mathcal{G}$ may be the Nisan--Wigderson generator applied to $f$), so the complexity of $\mathcal{G}$ is necessarily bounded from below by $\size(f)$.
Because the reconstruction argument only proves that $\mathcal{G}$ is pseudorandom against circuits of size $s \ll \size(f) \le \size(\mathcal{G})$, this proof template is doomed to fail in constructing a cryptographic hitting set generator.
To design cryptographic hitting set generators for strong circuit classes, we need to avoid this reconstructive approach.

One instance where the overhead from reconstruction can be avoided is found in the work of \textcite{AF22}.
They gave a subexponential-time algorithm to test polynomial identities that are written as constant-depth circuits.
The same algorithmic result was already obtained by the previously-mentioned work of \textcite{LST21}.
However, these two algorithms differ in the complexity of the hitting set generator underlying the algorithm: the generator of \cite{LST21} is reconstructive, whereas the generator of \cite{AF22} provably has smaller complexity than the circuits it hits.
Unlike most works in algebraic hardness-randomness, \textcite{AF22} do not design a generator that is based on evaluations of hard functions, but rather design the generator so that its annihilator ideal consists only of hard polynomials.
For a generator $\mathcal{G}$, its \emph{annihilator ideal}, denoted by $\Ann{\mathcal{G}}$, is the set of all polynomials $f$ such that $f \circ \mathcal{G} = 0$.
In principle, one could show that a generator hits a circuit class $\mathscr{C}$ by showing that all nonzero polynomials in the annihilator ideal are so complex that they lie outside the circuit class $\mathscr{C}$.
Proving a statement like this is necessary for the proof of correctness of a generator, but these statements are usually not the core thrust of the argument and only arise as a byproduct of the correctness proof.
As we will see, adopting this viewpoint will be useful for constructing further examples of generators that hit circuits more complex than the generator itself.

\subsection{Cryptographic Generators from Degree Bounds for Annihilating Polynomials} \label{subsec:annihilator degree bounds}

Although the question of constructing cryptographic hitting set generators in algebraic complexity is a recent one, prior work on degree bounds for annihilating polynomials leads, at least implicitly, to constructions of cryptographic generators.
\textcite[Theorem 12]{Kayal09} showed that over a field of characteristic zero, any annihilator of the $n+1$ polynomials
\begin{align*}
    g_1(\vec{x}) &= x_1^d - 1 \\
    & \vdotswithin{=} \\
    g_n(\vec{x}) &= x_n^d - 1 \\
    g_{n+1}(\vec{x}) &= x_1 + x_2 + \cdots + x_n - n
\end{align*}
must have degree at least $d^n$.
Equivalently, the corresponding generator $\mathcal{G} : \F^n \to \F^{n+1}$ given by $\mathcal{G}(\vec{x}) = (g_1(\vec{x}), \ldots, g_{n+1}(\vec{x}))$ hits all polynomials of degree less than $d^n$.
Because many restricted forms of arithmetic circuits, such as arithmetic formulas and branching programs, necessarily compute low-degree polynomials, this generator hits powerful classes of circuits.
For example, size-$s$ arithmetic branching programs cannot compute polynomials of degree greater than $s$, so this generator hits \emph{all} branching programs of size less than $d^n$.
This is a bit unusual. 
The best-known lower bounds against arithmetic branching programs are only quadratic \cite{CKSV22}, so intuition from the hardness versus randomness paradigm suggests we should only expect to have constructions of generators that hit branching programs of size up to $O(n^2)$.

The preceding example of \textcite{Kayal09} fits into the cryptographic regime, as the outputs are extremely simple to compute.
Taking $d = 2$, we obtain a generator where the first $n$ outputs can be computed by formulas of constant size, the last output can be computed by a formula of size $n$ and depth two, and yet the generator hits all branching programs of size less than $2^n$.
Although this generator only stretches its input by one field element in length, it is possible to improve the stretch of the generator by invoking $n^{1 - \eps}$ copies in parallel, each on a fresh set of $n^\eps$ variables. 
This results in a generator $\mathcal{G}' : \F^n \to \F^{n + n^{1 - \eps}}$ that stretches its input by $n^{1 - \eps}$ field elements.
Kayal's construction can also be modified to obtain a generator with similar parameters where \emph{all} outputs are computable by constant-size formulas.
To do this, we encode the final output $g_{n+1}$ with additional variables $y_2, \ldots, y_{n-1}$ via
\begin{align*}
    h_{n+1}(\vec{x}, \vec{y}) &= x_1 + x_2 - y_2 \\
    h_{n+2}(\vec{x}, \vec{y}) &= y_2 + x_3 - y_3 \\
    &\vdotswithin{=} \\
    h_{2n-1}(\vec{x}, \vec{y}) &= y_{n-1} + x_n - n.
\end{align*}
Kayal's lower bound on the degree of the annihilator extends to this variation on his example, so we already have constructions of generators where each output is computable by a formula of constant size, yet the generator itself hits polynomials of much higher complexity.

\subsection{The Ideal Proof System}

Not only is the problem of constructing cryptographic hitting set generators interesting in its own right, but it is also closely tied to open problems in algebraic proof complexity.
The central goal of propositional proof complexity is to understand the $\NP$ versus $\coNP$ problem via the complexity of the $\coNP$-complete unsatisfiability problem: given a boolean formula $\varphi$, accept $\varphi$ if and only if $\varphi$ is unsatisfiable.
The typical setting is to first fix a proof system, and then try to find families of boolean formulas $(\varphi_n)_{n \in \naturals}$ such that any proof $\pi_n$ of the unsatisfiability of $\varphi_n$ requires length that is super-polynomial in $n$.
Numerous proof systems based on different areas of mathematics, including logic, algebra, and geometry, have been considered.

Most relevant to our work are proof systems based on algebraic reasoning.
Without loss of generality, we may assume that our unsatisfiable boolean formula $\varphi$ is in 3CNF form.
If $\varphi$ is an $n$-variate 3CNF formula with $m$ clauses, there is a translation of $\varphi$ to a system $\mathcal{F}$ of $n+m$ polynomial equations such that $\mathcal{F}$ is satisfiable if and only if $\varphi$ is satisfiable.
This system $\mathcal{F}$ consists of the $n$ boolean axioms $x_i^2 - x_i = 0$, which force the variables $x_i$ to be $\bits$-valued, and for each clause of the form $(x_1 \oplus b_1) \lor (x_2 \oplus b_2) \lor (x_3 \oplus b_3)$, the trivariate equation
\[
    (x_1 - (1 - b_1))\, (x_2 - (1 - b_2))\, (x_3 - (1 - b_3)) = 0.
\]
It is easy to see that any boolean assignment to the $x$ variables satisfies this equation if and only if the same assignment satisfies the corresponding clause $(x_1 \oplus b_1) \lor (x_2 \oplus b_2) \lor (x_3 \oplus b_3)$.
Thus, if we want to prove that $\varphi$ is unsatisfiable, we can instead try to prove the unsatisfiability of the resulting system of polynomial equations $\mathcal{F}$.

How does one prove that a system of polynomial equations is unsatisfiable?
The answer comes from Hilbert's Nullstellensatz.
Over an algebraically closed field $\F$, the system of equations 
\[
    f_1(\vec{x}) = \cdots = f_m(\vec{x}) = 0 
\]
is unsatisfiable if and only if there are polynomials $g_1, \ldots, g_m \in \F[x_1,\ldots,x_n]$ such that
\[
    \sum_{i=1}^m f_i(\vec{x}) \, g_i(\vec{x}) = 1.
\]
We can take the polynomials $\set{g_1,\ldots,g_m}$ to be our proof of unsatisfiability.
Our goal, then, is to understand the complexity of the simplest refutation $\set{g_1,\ldots,g_m}$ of the system of equations $\mathcal{F}$.

The complexity of the proof depends on the choice of proof system.
The Nullstellensatz proof system \cite{BIKPP96,BIKPRS96} measures the length of the proof by the total number of monomials in the $g_i$.
A slightly stronger system is the Polynomial Calculus \cite{Razborov98,IPS99}, where we are allowed to write the derivation of $1$ from the $f_i$ in a line-by-line manner, but we still pay for the number of monomials written on every line.
Much stronger is the Ideal Proof System (IPS), introduced by \textcite{GP18}, which allows us to write the $g_i$ succinctly as arithmetic circuits.
Given that refutations in IPS are written as arithmetic circuits, one would hope that techniques for proving arithmetic circuit lower bounds will eventually lead to lower bounds for the IPS.

This hope has played out successfully in recent years.
\textcite{FSTW21} introduced two techniques to infer IPS lower bounds from arithmetic circuit lower bounds, and applied these techniques for restricted circuit classes to conclude unconditional lower bounds.
The first technique, known as the functional lower bound method, has been further refined and studied by \textcite{GHT22, HLT24}.
Here, the hard system of equations $\mathcal{F}$ typically consists of boolean axioms together with a sparse polynomial that corresponds to an instance of subset sum, possibly lifted with an appropriate gadget.\footnote{These hard instances are closely related to the previously-discussed example of \textcite{Kayal09} that leads to lower bounds on the degree of annihilating polynomials.}
This method is capable of proving strong lower bounds; for example, \textcite{HLT24} proved super-polynomial lower bounds on the size of constant-depth IPS refutations for such systems of equations.
However, current applications of the functional lower bound method are only able to prove lower bounds against refutations of individual degree $O(\log \log n)$.
This is to be expected, as functional lower bounds without the individual degree constraint, even for depth-four arithmetic circuits, would lead to new lower bounds in boolean complexity theory \cite{FKS16}.

The second technique introduced by \textcite{FSTW21} is based on hardness of multiples.
This method is also capable of proving strong lower bounds: later work by \textcite{AF22} proved super-polynomial lower bounds on the size of constant-depth IPS refutations for a system of equations related to matrix rank.
Unlike the functional lower bound method, the lower bounds produced from hardness of multiples do not require the IPS refutation to be of small individual degree.
However, the drawback to this method is that in order to prove a lower bound against IPS certificates computed by $\mathscr{C}$-circuits, there must be at least one polynomial $f$ in the hard instance $\mathcal{F}$ which itself cannot be computed efficiently by $\mathscr{C}$-circuits.
This limitation is inherent to the method, as the lower bound against IPS is derived from circuit lower bounds for one (or more) of the polynomials in the hard instance $\mathcal{F}$.

Strong conditional lower bounds for IPS are also known.
\textcite{AGHT24} proved super-polynomial lower bounds on the size of constant-free IPS refutations of the binary value principle, a system of equations that asserts a natural number $n$ given in binary is negative.
Their lower bound assumes the $\tau$-conjecture of Shub and Smale, which asserts that the straight-line program complexity of computing (a multiple of) the integer $n!$ is bounded from below by $\log^{\omega(1)}n$.
\textcite{ST21a} showed that under the assumption $\VP \neq \VNP$, there is a sequence of 3CNF formulas that require IPS refutations of super-polynomial size.
One drawback of this result is that the sequence of CNF formulas considered by \textcite{ST21a} may be satisfiable; if this were the case, then the lower bound becomes trivial, as the soundness of IPS implies that it cannot refute a system of satisfiable equations.

Despite the success so far in bringing techniques from arithmetic circuit complexity to bear on IPS, we still seem far from proving unconditional lower bounds for systems of polynomial equations that encode unsatisfiable 3CNF formulas, even against weak subsystems of IPS.
Current applications of the functional lower bound method are limited to proving lower bounds against refutations of low individual degree, and some formulations of the method provably cannot lead to lower bounds for boolean instances \cite{HLT24}.
The method of lower bounds for multiples can prove unconditional lower bounds against fragments of IPS with no restrictions on the individual degree of the refutation, but the method requires the hard instance to contain polynomials of large circuit complexity, which precludes its application to systems of equations that encode a 3CNF formula.
As a step towards proving IPS lower bounds for such systems, can we prove IPS lower bounds for any system of equations where each polynomial in the system can be described by an arithmetic formula of constant size?

\subsection{Hard Families of Simple Polynomials from Nullstellensatz Degree Bounds} \label{subsec:nullstellensatz degree bounds}

Just as prior work on annihilating polynomials led to simple constructions of cryptographic hitting set generators, existing work on lower bounds for Nullstellensatz degree---in particular, examples demonstrating the tightness of these bounds---easily leads to families of simple polynomials that are hard for fragments of IPS.
Recall that the Nullstellensatz says that if $f_1 = \cdots = f_m = 0$ is an unsatisfiable system of equations, then there are polynomials $g_1, \ldots, g_m$ such that $\sum_{i=1}^m f_i g_i = 1$.
Of particular relevance to our work are degree bounds for the Nullstellensatz: given the polynomials $f_i$, what is the smallest possible degree of the polynomials $g_i$ that witness the identity $\sum_{i=1}^n f_i g_i = 1$?
A long line of work \cite{Hermann1926,Brownawell87,CGH88,Kollar88,FG90,Sombra99,KPS01,Jelonek05} has established various bounds on the degrees and heights of such polynomials.
For example, we know that such polynomials $g_i$ can always be found with degree $\deg(g_i) \le d^n$, where $d = \max_{i} \deg(f_i)$ is the maximum degree of the $f_i$.
An example due to Masser and Philippon (see \cite{Brownawell87}) shows that bounds of this shape are tight.
In particular, if we take
\begin{align*}
    f_1(\vec{x}) &= x_1^d \\
    f_2(\vec{x}) &= x_1 - x_2^d \\
    & \vdotswithin{=} \\
    f_{n-1}(\vec{x}) &= x_{n-2} - x_{n-1}^d \\
    f_{n}(\vec{x}) &= 1 - x_{n-1} x_n^{d-1},
\end{align*}
then in any expression $\sum_{i=1}^n f_i g_i = 1$, the polynomial $g_1$ must satisfy $\deg(g_1) \ge d^n - d^{n-1}$.
This degree lower bound implies similar lower bounds on the size of IPS refutations of the above system when the refutation is written as an arithmetic formula or arithmetic branching program.
The lower bound follows from the fact that an arithmetic formula or branching program of size $s$ can only compute polynomials of degree at most $s$, so no formula or branching program of size less than $d^n - d^{n-1}$ can compute a refutation of this system.

At first glance, it appears that we have made progress towards IPS lower bounds for refuting 3CNF formulas.
The example of Masser and Philippon above gives us a system of polynomial equations, each of which can be encoded by a constant-size arithmetic formula, that requires large refutations in powerful subsystems of IPS.
Unfortunately, this line of reasoning cannot lead to IPS lower bounds for systems of polynomials that encode 3CNF formulas.
An unsatisfiable 3CNF formula on $n$ variables always admits a degree-$O(n)$ refutation \cite{GP18}, so degree bounds will not result in new lower bounds on the complexity of refuting 3CNF formulas.
To make progress towards proving lower bounds for IPS refutations of 3CNF formulas, we need techniques rooted in finer complexity measures than degree.

\subsection{Our Contributions}

We now describe our results.
As we have seen in \cref{subsec:annihilator degree bounds,subsec:nullstellensatz degree bounds}, prior work on degree bounds can be used to construct cryptographic hitting set generators and hard instances for the Ideal Proof System.
The main contribution of our work is a new construction of cryptographic hitting set generators whose correctness is based on arithmetic circuit lower bounds, not degree lower bounds.
As a corollary, we will obtain families of simple polynomials that are hard to refute in subsystems of the the Geometric Ideal Proof System, itself a restricted form of the Ideal Proof System \cite[Appendix B]{GP18}.
Although the theorem statements below already follow from prior work, the constructions appearing in our proofs do not, and we believe there is value in developing techniques in algebraic pseudorandomness and proof complexity that go beyond degree lower bounds.

Our first result is an explicit construction of a low-complexity hitting set generator that hits $\VAC^0$, the class of polynomials computed by constant-depth circuits of polynomial size.
Each output of our generator is computable in $\VNC^0$, i.e., can be computed by an arithmetic formula of constant size.

\begin{theorem}[see \cref{thm:vac0 generator}] \label{thm:vac0 main}
    Let $\F$ be a field of characteristic zero.\footnote{For technical reasons, we can only prove that our generator is pseudorandom over fields of characteristic zero or sufficiently large characteristic. See \Cref{sec:instantiate} for an explanation of the underlying issue.}
    There is a $\VNC^0$-computable hitting set generator with stretch $n^{0.99}$ that hits $\VAC^0$.
\end{theorem}

Under strong but reasonable hardness assumptions, the same techniques yield $\VNC^0$-computable hitting set generators that hit larger circuit classes.
Our first conditional construction yields a $\VNC^0$-computable generator that hits $\VF$, which corresponds to families of polynomials that are computable by polynomial-size formulas.
The correctness of this construction relies on the assumption that the border formula complexity of the determinant is super-polynomial.

\begin{theorem}[see \cref{thm:vf generator}] \label{thm:vf main}
    Let $\F$ be a field of characteristic zero.
    If the border formula complexity of the $n \times n$ determinant is $n^{\omega(1)}$, then there is a $\VNC^0$-computable hitting set generator that hits $\VF$.
\end{theorem}

Although it is commonly conjectured that the formula complexity of the determinant is $n^{\omega(1)}$, it is less clear whether we should expect the same lower bound to hold for border formulas, as border computation is poorly understood even in very weak settings.
Despite this gap in our understanding, we find it conceivable that the determinant does require border formulas of super-polynomial size.
The reason for this is that recent progress on arithmetic circuit lower bounds \cite{LST21,TLS22,KS22,KS23,FLST24,Forbes24} has relied on the use of complexity measures based on matrix rank.
Because matrix rank is lower semi-continuous, these lower bounds often directly imply lower bounds on border complexity.\footnote{See e.g.~\cite[Section 6.1]{AF22} for an explanation of how the lower bound of \textcite{LST21} implies a lower bound on border complexity.}
It is unclear if these methods will prove lower bounds against arithmetic formulas in the near future, and if they do, the hard polynomial may not be the determinant.
However, in light of the recent success of rank-based methods in proving lower bounds, coupled with the fact that many of these lower bounds apply to the determinant, it is not out of the question that when we manage to prove lower bounds against arithmetic formulas, the methods used will be rank-based and will apply to the determinant. 
In that case, we will have corresponding lower bounds on the border formula complexity of the determinant.

Our second conditional construction produces a $\VNC^0$-computable generator that hits polynomials that are computable by polynomial-size arithmetic branching programs, a class commonly known as $\VBP$.
This result assumes there is a family of polynomials that can be computed by polynomial-size arithmetic circuits but not by polynomial-size arithmetic branching programs.

\begin{theorem}[see \cref{thm:vbp generator}] \label{thm:vbp main}
    Let $\F$ be a field of characteristic zero.
    If there is a family of polynomials in $\VP$ that require arithmetic branching programs of super-polynomial size, then there is a $\VNC^0$-computable hitting set generator that hits $\VBP$.
\end{theorem}

As a corollary of our generator constructions, we obtain new lower bounds for subsystems of the Geometric Ideal Proof System of \textcite[Appendix B]{GP18}, a restricted form of the Ideal Proof System.
An easy observation, already made by \textcite{GKSS17}, shows that hitting set generators immediately yield hard instances for the Geometric Ideal Proof System.
The complexity of the generator upper bounds the complexity of the equations in the hard instance, and the hitting property of the generator is used to infer the lower bound against Geometric IPS.
Because we construct $\VNC^0$-computable generators against $\VAC^0$, we obtain a system of equations where every equation can be written as an arithmetic formula of constant size, but no Geometric IPS refutation can be computed by a circuit of constant depth and polynomial size.

\begin{theorem}[see \cref{thm:vac0 geom ips lb}]
    Let $\F$ be a field of characteristic zero.
    There is an explicit system of polynomial equations $\mathcal{F}_n$ such that (1) each equation in $\mathcal{F}_n$ depends on at most $3$ variables, (2) the system $\mathcal{F}_n$ can be refuted by Geometric IPS, and (3) any Geometric IPS refutation of $\mathcal{F}_n$ cannot be computed by a circuit of constant depth and polynomial size.
\end{theorem}

We also prove conditional lower bounds against Geometric IPS refutations computed by arithmetic formulas (\cref{thm:vf geom ips lb}) and arithmetic branching programs (\cref{thm:vbp geom ips lb}).
The conditions used to prove these lower bounds are the same ones used to construct $\VNC^0$-computable generators that hit arithmetic formulas and arithmetic branching programs, respectively.

Although our hard instances do not correspond to encodings of 3CNF formulas, we view these results as progress towards proving lower bounds for refutations of unsatisfiable boolean formulas.
As mentioned in \Cref{subsec:nullstellensatz degree bounds}, unsatisfiable 3CNF formulas on $n$ variables can always be refuted in degree $O(n)$, so any method of proving lower bounds on the complexity of refuting 3CNF's must tolerate the existence of low-degree refutations.
In all of our lower bounds for Geometric IPS, the hard systems of equations admit refutations of degree $n^{O(1)}$.
To the best of our knowledge, this is the first example of a system of equations where each equation can be expressed by an arithmetic formula of constant size, the system admits low-degree refutations, and the system is provably hard to refute in a nontrivial subsystem of IPS.

\subsection{Our Techniques}

We obtain our $\VNC^0$-computable hitting set generators through a transformation that takes a known hitting set generator $\mathcal{G}$ and compiles it into a related generator $\mathcal{G}'$ that has much lower complexity, yet retains the hitting property of $\mathcal{G}$.
A similar high-level approach was taken by \textcite{AIK06} to obtain $\NC^0$-computable one-way functions and pseudorandom generators in the boolean setting.
There are some superficial similarities between our work and theirs, since both works obtain the new generator $\mathcal{G}'$ through an encoding of the computation of $\mathcal{G}$.

However, one difference between our work and \cite{AIK06} is the set of requirements placed on the high-complexity generator $\mathcal{G}$.
\textcite{AIK06} require $\mathcal{G}$ to be of sufficiently low complexity (say, $\NC^1$-computable), but are agnostic to the particular choice of the generator $\mathcal{G}$.
In contrast, our work does not place a requirement on the complexity of the starting generator, but we are only able to handle generators $\mathcal{G}$ that are of the specific form
\[
    \mathcal{G} : (x_1, \ldots, x_n) \mapsto (x_1, \ldots, x_n, f(\vec{x})),
\]
where $f(\vec{x})$ is a polynomial.

We require the starting generator to be of this form because it provides sufficient algebraic structure to completely analyze the resulting low-complexity generator $\mathcal{G}'$.
To every generator $\mathcal{G} : \F^\ell \to \F^n$, one can associate its \emph{annihilator ideal} $\Ann{\mathcal{G}}$, the set of polynomials $h \in \F[z_1,\ldots,z_n]$ such that $h \circ \mathcal{G} = 0$, i.e., the composition of $h$ and $\mathcal{G}$ results in the identically zero polynomial.
To prove that a generator $\mathcal{G}$ hits a circuit class $\mathscr{C}$, it is both necessary and sufficient to show that every nonzero polynomial in the annihilator $\Ann{\mathcal{G}}$ cannot be computed within the resource bounds of $\mathscr{C}$.

For some generators, we can completely characterize the ideal $\Ann{\mathcal{G}}$, which a useful first step towards proving lower bounds for $\Ann{\mathcal{G}}$.
In the case where $\mathcal{G}$ is of the form $\mathcal{G}(\vec{x}) = (x_1,\ldots,x_n, f(\vec{x}))$, it is easy to show that the annihilator ideal is given by
\[
    \Ann{\mathcal{G}} = \abr{z_{n+1} - f(z_1,\ldots,z_n)} \subseteq \F[z_1,\ldots,z_{n+1}].
\]
That is, every annihilator of $\mathcal{G}$ is a multiple of the polynomial $z_{n+1} - f(z_1,\ldots,z_n)$.
Because this ideal is generated by a single polynomial (i.e., is a principal ideal), we can leverage existing techniques on polynomial factorization to prove lower bounds for $\Ann{\mathcal{G}}$.
The argument proceeds by contradiction: if there is a small $\mathscr{C}$-circuit that computes an element of $\Ann{\mathcal{G}}$, and if $\mathscr{C}$-circuits are polynomially-closed under factorization, then there is a small $\mathscr{C}$-circuit for $z_{n+1} - f(z_1,\ldots,z_n)$.
Thus, if the circuit class $\mathscr{C}$ is closed under factorization (as are $\VP$ \cite{Kaltofen87} and $\VBP$ \cite{ST21c}), then to prove lower bounds for $\Ann{\mathcal{G}}$, it suffices to prove a lower bound on $f(z_1,\ldots,z_n)$.
This is a much easier task, since we only have to reason about a single polynomial $f$ and not an arbitrary multiple of $f$.
In some cases, the circuit class $\mathscr{C}$ is not known to be polynomially-closed under factorization.
Despite this, we can still make use of results that show for a \emph{specific choice} of $f$, lower bounds on the $\mathscr{C}$-circuit complexity of $f$ imply lower bounds for all multiples of $f$.

So far, we have seen that it is possible to completely understand generators of the form $\mathcal{G}(x_1,\ldots,x_n) = (x_1, \ldots, x_n, f(\vec{x}))$.
On its own, this generator has no hope of producing a generator that is of lower complexity than the circuits it hits, since there is always an annihilator of complexity comparable to $f$.

To obtain a generator of lower complexity, we replace the single output $f(\vec{x})$ by a sequence of $s+1$ outputs that encode a size-$s$ circuit $\Phi$ that computes $f(\vec{x})$.\footnote{This transformation is similar to the reduction of Circuit-SAT to 3CNF-SAT and was recently used by \textcite{Grochow23} to study the PIT instances that arise from verification of IPS refutations.}
For each internal gate of $\Phi$, we introduce a fresh variable $y_i$, and we include the polynomial 
\[
    y_i - y_j \odot y_k 
\]
in the output of the generator, where $\odot \in \set{+, \times}$ is the operation labeling the $i$\ts{th} gate and the children of the $i$\ts{th} gate are gates $j$ and $k$.\footnote{If the children of the $i$\ts{th} gate are not internal gates, we use a different polynomial, but this is a technical point that does not meaningfully impact the overview here. See \cref{def:local encoding} for the precise definition.}
We also include $y_s$ as an output so that the generator stretches its $n+s$ inputs to $n+s+1$ outputs.
It is clear from the definition that each output of this new generator $\mathcal{G}'$ can be computed by an arithmetic formula of constant size, but it is not obvious why $\mathcal{G}'$ preserves any pseudorandom properties the original generator $\mathcal{G}$ had.

To show that $\mathcal{G}'$ is pseudorandom, we follow the approach suggested above: we determine the annihilator ideal $\Ann{\mathcal{G}'}$ and subsequently prove lower bounds on the complexity of all nonzero polynomials in $\Ann{\mathcal{G}'}$.
As we saw, the annihilator ideal $\Ann{\mathcal{G}}$ of the starting generator was generated by $z_{n+1} - f(z_1,\ldots,z_n)$, so we could infer lower bounds on $\Ann{\mathcal{G}}$ from lower bounds for $f$ and its multiples.
For $\mathcal{G}'$, the annihilator ideal $\Ann{\mathcal{G}'}$ is again principal and is generated by a polynomial of the form
\[
    h(\vec{z}) = z_{n+s+1} - f(z_1,\ldots,z_n) + g(\vec{z}),
\]
where $g$ is a structured polynomial that should be thought of as an error term for the purposes of this overview.
To prove that $\mathcal{G}'$ is pseudorandom, it suffices to prove lower bounds on the complexity of multiples of $h$.
Because $h$ is close to $f$, one could hope that lower bounds for multiples of $f$ imply comparable lower bounds for multiples of $h$.
This is precisely the route we follow to show that $\mathcal{G}'$ is pseudorandom.
By instantiating the construction of $\mathcal{G}'$ with a polynomial $f$ whose multiples are hard for a circuit class $\mathscr{C}$ (possibly under hardness assumptions), we obtain a generator that is computable in $\VNC^0$ yet hits the much larger class $\mathscr{C}$.

We encourage readers who want to see a concrete example of this new generator $\mathcal{G}'$ to skip ahead to \cref{ex:local encoding}. 
There, we describe and analyze the generator $\mathcal{G}'$ when the polynomial $f$ is taken to be $f(x_1,x_2) = x_1^2 - x_2^2$. 
This polynomial is too simple to yield a generator with useful pseudorandom properties, but its simplicity allows us to explicitly write down the resulting generator and its annihilator ideal.

\subsection{Organization}

The rest of this paper is organized as follows.
We review preliminary material in \cref{sec:prelim}.
We then begin in \cref{sec:generator construction}, where we describe a general construction of $\VNC^0$-computable hitting set generators.
\cref{sec:instantiate} gives three concrete instantiations of this generator construction.
\cref{sec:geom ips} uses our generators to obtain lower bounds for the Geometric Ideal Proof System.
Finally, we conclude in \cref{sec:conclusion} with some open questions.

\section{Preliminaries} \label{sec:prelim}

For a natural number $n \in \naturals$, we write $[n] \coloneqq \set{1,2,\ldots,n}$.
We abbreviate a vector of variables $(x_1,\ldots,x_n)$ as $\vec{x}$.
For a field $\F$, we write $\F[\vec{x}]$ and $\F(\vec{x})$ for the polynomial ring and field of rational functions, respectively, over $\F$ in the variables $x_1, \ldots, x_n$.
A \emph{polynomial map} $\mathcal{G}$ is a tuple of polynomials $(g_1(\vec{x}), \ldots, g_n(\vec{x})) \in \F[x_1,\ldots,x_{\ell}]^n$, which defines a map $\mathcal{G} : \F^\ell \to \F^n$ via $\mathcal{G}(\vec{\alpha}) = (g_1(\vec{\alpha}), \ldots, g_n(\vec{\alpha}))$.
We often abuse notation and refer to the induced map $\mathcal{G} : \F^\ell \to \F^n$ as a polynomial map; the two objects are equivalent over infinite fields, but when $\F$ is finite, different tuples of polynomials may induce the same map $\F^\ell \to \F^n$.
For a collection of polynomials $f_1, \ldots, f_m \in \F[\vec{x}]$, we write $\abr{f_1,\ldots,f_m}$ for the ideal of $\F[\vec{x}]$ generated by $f_1, \ldots, f_m$.

\subsection{Arithmetic Circuit Complexity}

This subsection recalls basic notions of arithmetic circuit complexity.
For more on this topic, we refer the reader to the surveys of \textcite{SY10} and \textcite{Saptharishi-survey}.

\begin{definition}[Arithmetic circuits]
    Let $\F$ be a field and let $\F[\vec{x}]$ be the polynomial ring over $\F$ in $n$ variables.
    An \emph{arithmetic circuit $\Phi$} is a directed acyclic graph equipped with the following data.
    The vertices of in-degree zero are called \emph{input gates} and are labeled by either a variable $x_i$ or a constant $\alpha \in \F$.
    Vertices of positive in-degree are called \emph{internal gates} and are labeled by either addition or multiplication.
    Each vertex $v$ of $\Phi$ computes a polynomial $f_v(\vec{x}) \in \F[\vec{x}]$ in the natural way.
    If a vertex $v$ has out-degree zero, we call $v$ an \emph{output gate} and say that the circuit $\Phi$ computes $f_v(\vec{x})$.

    We measure the \emph{size} of $\Phi$ by the number of internal gates in the circuit.
    The \emph{depth} of a circuit is the length of the longest path from any input gate to any output gate.
\end{definition}

We will sometimes insist on working with circuits whose internal gates have fan-in two.
A circuit with internal gates of unbounded fan-in can be simulated by a circuit where every internal gate has fan-in two by replacing each gate of large fan-in with a binary tree of fan-in two gates.
This has the potential to increase the size of the circuit from $s$ to $O(s^2)$, which will be negligible for our purposes.

Next, we define arithmetic formulas, which are arithmetic circuits whose the underlying graph is a tree.

\begin{definition}[Arithmetic formulas]
    An \emph{arithmetic formula} is an arithmetic circuit whose underlying graph is a tree.
    Equivalently, an arithmetic formula is an arithmetic circuit in which every gate has out-degree at most one.
\end{definition}

We also need the notion of arithmetic branching programs, whose expressive power lies somewhere between that of formulas and circuits.

\begin{definition}[Arithmetic branching programs]
    An \emph{arithmetic branching program} is a layered directed acyclic graph $G = (V, E)$ with a single source vertex $s$ and a single sink vertex $t$.
    The fact that $G$ is layered means that there is a partition $V = V_0 \sqcup V_1 \sqcup \cdots \sqcup V_d$ such that $V_0 = \set{s}$, $V_d = \set{t}$, and every edge of $G$ is between vertices in $V_{i-1}$ and $V_i$ for some $i \in [d]$.
    Each edge $e$ of $G$ is labeled by an affine linear polynomial $\ell_e(\vec{x}) \in \F[\vec{x}]$.
    The branching program computes the polynomial 
    \[
        \sum_{P : s \rightsquigarrow t} \prod_{e \in P} \ell_e(\vec{x}),
    \]
    where the sum is over all $(s,t)$-paths in $G$.
    We measure the \emph{size} of the branching program by $|V|$, the total number of vertices.
\end{definition}

Having defined our model of computation, we now define the objects we are interested in computing.
These are families of polynomials whose degree is polynomially-bounded; such families of polynomials are called $p$-families.

\begin{definition}[$p$-families]
    Let $\F$ be a field and let $f = (f_n)_{n \in \naturals}$ be a sequence of polynomials with coefficients in $\F$.
    We say that $f$ is a \emph{$p$-family} if $\deg(f_n)$ is polynomially-bounded as a function of $n$.
\end{definition}

We now define the complexity classes we will be interested in throughout this work.
Although the definitions of these classes depend on the underlying field $\F$, we suppress this dependence for the sake of readability. 
The field $\F$ will always be clear from context.

\begin{definition}[Complexity classes]
    Let $\F$ be a field.
    \begin{enumerate}
        \item 
            The class $\VP$ consists of all $p$-families $(f_n)_{n \in \naturals}$ over $\F$ such that $f_n$ can be computed by an arithmetic circuit of size $n^{O(1)}$.
        \item 
            The class $\VBP$ consists of all $p$-families $(f_n)_{n \in \naturals}$ over $\F$ such that $f_n$ can be computed by an arithmetic branching program of size $n^{O(1)}$.
        \item 
            The class $\VF$ consists of all $p$-families $(f_n)_{n \in \naturals}$ over $\F$ such that $f_n$ can be computed by an arithmetic formula of size $n^{O(1)}$.
        \item 
            The class $\VAC^0$ consists of all $p$-families $(f_n)_{n \in \naturals}$ over $\F$ such that $f_n$ can be computed by an arithmetic circuit of size $n^{O(1)}$ and depth $O(1)$.
        \item 
            The class $\VNC^0$ consists of all $p$-families $(f_n)_{n \in \naturals}$ over $\F$ such that $f_n$ can be computed by an arithmetic circuit of size $O(1)$. \qedhere
    \end{enumerate}
\end{definition}

In addition to the standard notion of computing a polynomial via an arithmetic circuit, we will occasionally need to refer to border complexity, which corresponds to a notion of approximate computation using arithmetic circuits.

\begin{definition}[Border complexity] \label{def:border complexity}
    Let $\F$ be a field an $\eps$ be an indeterminate.
    Let $f(\vec{x}) \in \F[\vec{x}]$ be a polynomial.
    We say that a circuit $\Phi$ over $\F(\eps)$ \emph{approximately computes} $f(\vec{x})$ if $\Phi$ computes a polynomial of the form 
    \[
        f(\vec{x}) + \eps \cdot g(\vec{x}, \eps),
    \]
    where $g(\vec{x}, \eps) \in \F[\vec{x}, \eps]$ is a polynomial in $x_1,\ldots,x_n$ and $\eps$.
    We will abbreviate this by saying that $\Phi$ computes $f(\vec{x}) + O(\eps)$.
    If $\Phi$ has size $s$, we say that the \emph{border complexity} of $f$ is bounded by $s$.
\end{definition}

Over fields of characteristic zero, such as the complex numbers, one should interpret a circuit $\Phi$ that approximately computes $f$ as a circuit that computes $f$ in the limit as $\eps$ tends towards zero.
A potentially more general definition of approximate computation over $\complexes$ would be to consider a sequence of circuits $(\Phi_{m})_{m \in \naturals}$ that computes a sequence of polynomials $(f_m)_{m \in \naturals}$ such that $f = \lim_{m \to \infty} f_m$, without requiring a uniform description of the circuits $(\Phi_m)_{m \in \naturals}$ as in \cref{def:border complexity}.
However, a result of \textcite{Alder84} shows that these two notions of computation coincide over $\complexes$, so we may work with \cref{def:border complexity} without any loss in generality.

An alternate description of the limit-based definition is that the target polynomial $f$ lies in the Euclidean closure of the set of polynomials of complexity $s$.
By replacing the Euclidean topology with the Zariski topology, one obtains a notion of approximate computation that extends to algebraically closed fields of arbitrary characteristic.
As in the preceding paragraph, \textcite{Alder84} showed that this notion of approximate computation agrees with \cref{def:border complexity}.

Naturally, one can define complexity classes in terms of border complexity.

\begin{definition}[Border complexity classes]
    Let $\F$ be a field and let $f = (f_n)_{n \in \naturals}$ be a $p$-family.
    For a complexity class $\mathscr{C}$, we say that $f \in \overline{\mathscr{C}}$ if there is a $p$-family $g = (g_n)_{n \in \naturals}$ over $\F(\eps)$ such that $g \in \mathscr{C}$ and for all $n$, we have $g_n(\vec{x}) = f_n(\vec{x}) + O(\eps)$.
\end{definition}

For example, a $p$-family $(f_n)_{n \in \naturals}$ is in $\overline{\VP}$ if there is a sequence of circuits of size $n^{O(1)}$ over $\F(\eps)$ that compute $f_n + O(\eps)$.

\subsection{Polynomial Identity Testing}

Many deterministic algorithms for polynomial identity testing are obtained by explicitly constructing a hitting set generator.
A hitting set generator, defined below, is a pseudorandom object that plays a role analogous to pseudorandom (and hitting set) generators in boolean derandomization, stretching a short seed of truly random field elements into a longer string of pseudorandom field elements.

\begin{definition}[Hitting set generator] \label{def:hsg}
    Let $\F$ be a field and let $\mathscr{C}$ and $\mathscr{D}$ be classes of $p$-families over $\F$.
    Let $\mathcal{G} = (\mathcal{G}_n : \F^{\ell(n)} \to \F^n)_{n \in \naturals}$ be a sequence of polynomial maps, where $\mathcal{G}_n = (g_{n,1}, \ldots, g_{n,n}) \in \F[\vec{y}]^n$.
    We say that $\mathcal{G}$ is a \emph{$\mathscr{D}$-computable hitting set generator for $\mathscr{C}$} if the following conditions hold.
    \begin{enumerate}
        \item
            Every $p$-family of the form $g = (g_{n, i_n})_{n \in \naturals}$ is an element of $\mathscr{D}$.
        \item
            For every $p$-family $f = (f_n)_{n \in \naturals} \in \mathscr{C}$ and for all sufficiently large $n$, we have $(f_n \circ \mathcal{G}_{m(n)})(\vec{y}) = 0$ if and only if $f_n(\vec{x}) = 0$, where $m(n)$ is the number of variables appearing in the polynomial $f$.
    \end{enumerate}
    The \emph{seed length} of $\mathcal{G}$ is $\ell(n)$.
    The \emph{stretch} of $\mathcal{G}$ is $n - \ell(n)$.
    The \emph{degree} of $\mathcal{G}$ is $d(n) \coloneqq \max_{i \in [n]} \deg(g_{n, i})$.
\end{definition}

\begin{remark}[Padding] \label{rem:padding}
    In \cref{def:hsg}, we adopt the convention that the $n$\ts{th} map in the sequence $(\mathcal{G}_n : \F^{\ell(n)} \to \F^n)_{n \in \naturals}$ has output length $n$ largely as a matter of simplicity.
    When we construct generators later in \cref{sec:instantiate}, our constructions will naturally lead to sequences of maps where the $n$\ts{th} map has output length $m(n)$ for a polynomially-bounded function $m(n)$.
    To obtain a generator with output length $m(n) + q < m(n+1)$, we pad the input and output of the generator with $q - m(n)$ fresh variables.
    This preserves the $\VNC^0$-computability of the generator and does not impact the pseudorandom properties of the generator.
\end{remark}

For the sake of readability, we may refer to a single polynomial map $\mathcal{G} : \F^\ell \to \F^n$ as a hitting set generator.
This is done with the understanding that there is an underlying sequence of polynomial maps $(\mathcal{G}_n : \F^{\ell(n)} \to \F^n)_{n \in \naturals}$ and we have implicitly fixed a choice of $n$.

An explicit construction of a hitting set generator $\mathcal{G}$ for a circuit class $\mathscr{C}$ immediately yields an algorithm to test $\mathscr{C}$-circuits: given a polynomial $f \in \mathscr{C}$, test the composed polynomial $f \circ \mathcal{G}$ by brute force.
By the Schwartz--Zippel lemma, testing $f \circ \mathcal{G}$ can be done with $(\deg(f) \deg(\mathcal{G})+1)^\ell$ evaluations.
If $\ell$ is small compared to $n$ and $\deg(\mathcal{G})$ is not large, this yields an improvement over the brute-force identity testing algorithm, which evaluates $f$ at $(\deg(f) + 1)^n$ points.

One means of studying the pseudorandom properties of a candidate generator $\mathcal{G}$ is by studying its annihilator ideal $\Ann{\mathcal{G}}$, which we now define.

\begin{definition}[Annihilator ideal]
    Let $\mathcal{G} : \F^\ell \to \F^n$ be a polynomial map.
    The \emph{annihilator ideal} of $\mathcal{G}$, denoted by $\Ann{\mathcal{G}}$, is the set of all polynomials that vanish when composed with $\mathcal{G}$.
    Formally, 
    \[
        \Ann{\mathcal{G}} \coloneqq \set{f \in \F[x_1,\ldots,x_n] : f(g_1(\vec{y}), \ldots, g_n(\vec{y})) = 0}. \qedhere
    \]
\end{definition}

To understand the relation between the pseudorandom properties of a candidate generator $\mathcal{G}$ and its annihilator ideal $\Ann{\mathcal{G}}$, suppose we want to show that $\mathcal{G}$ hits some circuit class $\mathscr{C}$.
The statement ``$\mathcal{G}$ hits $\mathscr{C}$'' is equivalent to the containment 
\[
    \mathscr{C} \cap \Ann{\mathcal{G}} \subseteq \set{0}.
\]
This containment should be interpreted as saying that every nonzero polynomial in $\Ann{\mathcal{G}}$ is too complex to be computed within the resource bounds of the circuit class $\mathscr{C}$.
Thus, to prove that a candidate generator $\mathcal{G}$ hits a circuit class $\mathscr{C}$, one can instead prove lower bounds on every nonzero polynomial in $\Ann{\mathcal{G}}$.
This is the route we will take to prove the correctness of our generator.

We formalize the preceding discussion in the following lemma.

\begin{lemma} \label{lem:hitting vs annihilators}
    Let $\F$ be a field and let $\mathscr{C}$ be a class of $p$-families over $\F$.
    Let $\mathcal{G} = (\mathcal{G}_n : \F^{\ell(n)} \to \F^n)_{n \in \naturals}$ be a sequence of polynomials maps.
    Then $\mathcal{G}$ is a hitting set generator for $\mathscr{C}$ if and only if for every $p$-family $f = (f_n)_{n \in \naturals}$ where $f_n \in \Ann{\mathcal{G}_n} \setminus \set{0}$, we have $f \notin \mathscr{C}$.
\end{lemma}

\subsection{The Ideal Proof System}

In this subsection, we recall the Ideal Proof System (IPS) of \textcite{GP18}.
This is a proof system that uses algebraic reasoning to prove that systems of polynomial equations do not admit a solution.
In propositional proof complexity, the boolean axioms $x_i^2 - x_i = 0$ are often incorporated into the definition of a proof system, as they naturally appear when refuting encodings of unsatisfiable boolean formulas.
We will consider the IPS more generally as a proof system for arbitrary systems of polynomial equations, and so we adopt a definition that does not include the boolean axioms by default.

\begin{definition}[Ideal Proof System] \label{def:ips}
    Let $f_1, \ldots, f_m \in \F[x_1, \ldots, x_n]$ be polynomials such that the system of equations $f_1 = \cdots = f_m = 0$ has no solution over $\overline{\F}$, the algebraic closure of $\F$.
    An \emph{Ideal Proof System (IPS) refutation} of $(f_1, \ldots, f_m)$ is a polynomial $r \in \F[x_1, \ldots, x_n, z_1, \ldots, z_m]$ such that
    \begin{enumerate}
        \item
            $r(x_1, \ldots, x_n, f_1(\vec{x}), \ldots, f_m(\vec{x})) = 1$, and
        \item
            $r(x_1, \ldots, x_n, 0, \ldots, 0) = 0$. \qedhere
    \end{enumerate}
\end{definition}

For our purposes, we will be interested in the Geometric Ideal Proof System \cite[Appendix B]{GP18}, a natural subsystem of the IPS.

\begin{definition}[Geometric Ideal Proof System] \label{def:geometric ips}
    Let $f_1, \ldots, f_m \in \F[x_1, \ldots, x_n]$ be polynomials such that the system of equations $f_1 = \cdots = f_m = 0$ has no solution over $\overline{\F}$, the algebraic closure of $\F$.
    A \emph{Geometric Ideal Proof System refutation} of $(f_1, \ldots, f_m)$ is a polynomial $r \in \F[z_1, \ldots, z_m]$ such that
    \begin{enumerate}
        \item
            $r(f_1(\vec{x}), \ldots, f_m(\vec{x})) = 0$, and
        \item
            $r(0, \ldots, 0) = 1$. \qedhere
    \end{enumerate}
\end{definition}

Geometric IPS refutations have a natural interpretation from the viewpoint of algebraic geometry.
Consider the polynomial map $\mathbf{f} : \F^n \to \F^m$ given by $(x_1, \ldots, x_n) \mapsto (f_1(\vec{x}), \ldots, f_m(\vec{x}))$.
Suppose $r \in \F[z_1,\ldots,z_m]$ is a Geometric IPS refutation of $\mathbf{f}$.
The condition $r(f_1(\vec{x}), \ldots, f_m(\vec{x})) = 0$ implies the image of the map $\mathbf{f}$ lies in the \emph{variety} defined by $r$, denoted $\mathbf{V}(r)$, which is the set of all points $\vec{\alpha} \in \F^m$ such that $r(\vec{\alpha}) = 0$.

On the other hand, the second condition $r(0,\ldots,0) = 1$ implies that $(0,\ldots,0)$ is not in the variety $\mathbf{V}(r)$ defined by $r$.
Thus, the hypersurface $\V(r)$ serves as a geometric certificate that the point $(0,\ldots,0)$ lies outside the image $\im(\mathbf{f})$ of the map $\mathbf{f}$.
One can go further, noting that $\V(r)$ is a certificate that $(0,\ldots,0)$ lies outside the closure of $\im(\mathbf{f})$ in the Zariski topology, but we will not dwell on this.

For a polynomial map $\mathbf{f}$, one can view a Geometric IPS refutation $r$ as an annihilator of $\mathbf{f}$ that is additionally required to have a nonzero constant term.
The definition of Geometric IPS specifies that the constant term is 1, but the precise constant is not important.
We further explore the connection between annihilators and Geometric IPS refutations in \cref{sec:geom ips}.

Our focus will be on restricted subsystems of IPS and Geometric IPS.
For a complexity class $\mathscr{C}$, we use $\mathscr{C}$-IPS to refer to IPS with the additional restriction that the refutations are computable in $\mathscr{C}$, and likewise for Geometric $\mathscr{C}$-IPS.
The formal definition appears below.

\begin{definition}[$\mathscr{C}$-IPS, Geometric $\mathscr{C}$-IPS]
    Let $\mathcal{F} = (\mathcal{F}_n)_{n \in \naturals}$ be a family of systems of polynomial equations.
    Let $\mathscr{C}$ be a complexity class.
    We say that \emph{$\mathcal{F}$ can be refuted by $\mathscr{C}$-IPS} if there is a $p$-family $(r_n)_{n \in \naturals} \in \mathscr{C}$ such that $r_n$ is an IPS refutation of $\mathcal{F}_n$.
    If $r_n$ is a Geometric IPS refutation of $\mathcal{F}_n$, we say that \emph{$\mathcal{F}$ can be refuted by Geometric $\mathscr{C}$-IPS}.
\end{definition}

As mentioned in the introduction, we will be particularly interested in proving lower bounds for $\mathscr{C}$-IPS where the hard system of polynomial equations can be computed in some weaker class $\mathscr{D} \subsetneq \mathscr{C}$.
We formalize this with the notion of $\mathscr{D}$-equations below.

\begin{definition}[$\mathscr{D}$-equations] \label{def:d-equation}
    Let $\mathcal{F} = (\mathcal{F}_n)_{n \in \naturals}$ be a family of systems of polynomial equations, where
    \[
        \mathcal{F}_n = \set{f_{n,1}(\vec{x}) = 0, \ldots, f_{n, m(n)}(\vec{x}) = 0}.
    \]
    Let $\mathscr{D}$ be a complexity class.
    We say that $\mathcal{F}$ is a \emph{family of $\mathscr{D}$-equations} if every $p$-family of the form $(f_{n, i_n})_{n \in \naturals}$ is in $\mathscr{D}$.
\end{definition}

\subsection{The Jacobian Criterion}

We will need to compute the transcendence degree of sets of polynomials, which we can do over fields of characteristic zero using the Jacobian Criterion.
To state the Jacobian Criterion, we first need to define the Jacobian of a set of polynomials.

\begin{definition}[Jacobian]
    Let $f_1, \ldots, f_m \in \F[x_1,\ldots,x_n]$.
    The \emph{Jacobian} of $f_1,\ldots,f_m$, denoted $\Jac{f_1,\ldots,f_m}$, is the $m \times n$ matrix whose $(i,j)$ entry is given by
    \[
        \Jac{f_1,\ldots,f_m}_{i,j} \coloneqq \pd{f_i}{x_j}. \qedhere
    \]
\end{definition}

We now state the Jacobian Criterion, which gives a precise characterization of transcendence degree over fields of characteristic zero.

\begin{theorem}[Jacobian Criterion \cite{Jacobi1841}, see \cite{ER93}] \label{thm:jacobian criterion}
    Let $\F$ be a field of characteristic zero.
    Let $f_1,\ldots,f_m \in \F[x_1,\ldots,x_n]$.
    Then the transcendence degree of $\set{f_1,\ldots,f_m}$ satisfies
    \[
        \trdeg_{\F}(f_1,\ldots,f_m) = \rank_{\F(\vec{x})} \Jac{f_1,\ldots,f_m}.
    \]
\end{theorem}

As the following proposition shows, the rank of the Jacobian still provides a lower bound on the transcendence degree of a set of polynomials when working over a field of positive characteristic.

\begin{proposition}[{see \cite[Section 3]{DGW09}}] \label{prop:jacobian arbitrary field}
    Let $\F$ be an arbitrary field.
    Let $f_1, \ldots, f_m \in \F[x_1,\ldots,x_n]$.
    Then the transcendence degree of $\set{f_1,\ldots,f_m}$ satisfies
    \[
        \trdeg_{\F}(f_1,\ldots,f_m) \ge \rank_{\F(\vec{x})} \Jac{f_1,\ldots,f_m}.
    \]
\end{proposition}

\subsection{The Resultant}

We now define the resultant, a useful tool in polynomial factorization and elimination theory.

\begin{definition}[Resultant] \label{def:resultant}
    Let $\F$ be a field.
	Let $f(x) = \sum_{i=0}^n f_i x^i$ and $g(x) = \sum_{i=0}^m g_i x^i$ be univariate polynomials in $\F[x]$ of degrees $n$ and $m$, respectively.
    The \emph{resultant} of $f$ and $g$, denoted by $\res(f,g)$, is given by
	\[
		\res(f,g) \coloneqq \det 
        \begin{pmatrix}
			f_n & & & & g_m & & & \\
			f_{n-1} & f_n & & & g_{m-1} & g_m & & \\
			\vdots & \vdots & \ddots & & \vdots & \vdots & \ddots & \\
			\vdots & \vdots & & f_n & g_1 & \vdots & & g_m \\
			f_0 & \vdots & & f_{n-1} & g_0 & \vdots & & \vdots \\
			& f_0 & & \vdots & & g_0 & & \vdots \\
			& & \ddots & \vdots & & & \ddots & \vdots \\
			& & & f_0 & & & & g_0 
		\end{pmatrix},
	\]
    where the matrix above is an $(n+m) \times (n+m)$ matrix whose first $m$ columns are formed from the coefficients of $f$ and whose last $n$ columns are formed from the coefficients of $g$.
\end{definition}

As the following lemma shows, the resultant can be used to check if two polynomials share a common factor.

\begin{lemma}[{see, e.g., \cite[Corollary 6.17]{vzGG13}}] \label{lem:resultant common factor}
    Let $\F$ be a field and let $f, g \in \F[x]$.
    Then $\res(f,g) = 0$ if and only if $f$ and $g$ have a nontrivial common factor in $\F[x]$.
\end{lemma}

We will also make use of resultants of multivariate polynomials.
If $f, g \in \F[x_1, \ldots, x_n, y]$ are multivariate polynomials, we may regard them instead as univariate polynomials in $y$ with coefficients that are themselves polynomials in $x_1,\ldots,x_n$.
We write $\res_y(f(\vec{x},y), g(\vec{x}, y))$ for the corresponding resultant, which is a polynomial in $x_1, \ldots, x_n$.
The resultant $\res_y(f, g)$ can be used to test if $f$ and $g$ share a common factor.
However, this resultant only detects common factors that depend on the variable $y$.

\begin{lemma}[{see, e.g., \cite[Corollary 6.20]{vzGG13}}] \label{lem:resultant common factor multivariate}
    Let $\F$ be a field and let $f, g \in \F[x_1, \ldots, x_n, y]$.
    Then $\res_y(f,g) = 0$ if and only if $f$ and $g$ have a nontrivial common factor $h \in \F[x_1, \ldots, x_n, y]$ that depends on the variable $y$.
\end{lemma}

Although it is immediate from the definition that $\res_y(f, g)$ does not depend on the variable $y$, the following proposition says that the resultant eliminates $y$ from $f$ and $g$ in a structured manner.

\begin{proposition}[{see, e.g., \cite[Corollary 6.21]{vzGG13}}] \label{lem:resultant elimination}
    Let $\F$ be a field and let $f, g \in \F[x_1, \ldots, x_n, y]$.
    Then the resultant $\res_y(f(\vec{x}, y), g(\vec{x}, y))$ is an element of the elimination ideal $\abr{f, g} \cap \F[x_1, \ldots, x_n]$.
\end{proposition}

\section{A $\VNC^0$-Computable Generator} \label{sec:generator construction}

In this section, we describe and analyze a construction of a $\VNC^0$-computable hitting set generator $\mathcal{G}$.
The generator $\mathcal{G}$ is defined using an arithmetic circuit $\Phi$ that computes a polynomial $f(\vec{x})$.
We will obtain a complete description of the annihilator ideal $\Ann{\mathcal{G}}$: the ideal $\Ann{\mathcal{G}}$ is principal and is generated by a polynomial that is closely related to the polynomial $f(\vec{x})$ used to define the map $\mathcal{G}$.
This structure allows us to infer pseudorandom properties of $\mathcal{G}$ from hardness of multiples of $f(\vec{x})$.

\subsection{Local Encodings of Circuits}

We begin by describing the construction of our generator.

\begin{definition}[Local encoding] \label{def:local encoding} \label{cons:generator}
    Let $\Phi$ be an $n$-variate arithmetic circuit of fan-in two and size $s$.
    Let $\vec{\alpha} \in \F^n$ and let $\beta \in \F$.
    The \emph{local encoding of $\Phi(\vec{\alpha}) = \beta$} is the polynomial map $\mathcal{G} : \F^{n+s} \to \F^{n+s+1}$ defined by 
    \[
        \mathcal{G}(x_1,\ldots,x_n,y_1,\ldots,y_s) \coloneqq (\mathcal{G}_{\text{input}}(\vec{x}, \vec{y}), \mathcal{G}_{\text{internal}}(\vec{x},\vec{y}), \mathcal{G}_{\text{output}}(\vec{x}, \vec{y})),
    \]
    where each of the polynomial maps $\mathcal{G}_{\text{input}}$, $\mathcal{G}_{\text{internal}}$, and $\mathcal{G}_{\text{output}}$ are defined below.
    \begin{enumerate}
        \item 
            The polynomial map $\mathcal{G}_{\text{input}} : \F^{n+s} \to \F^n$ is given by
            \[
                \mathcal{G}_{\mathrm{input}}(\vec{x}, \vec{y}) \coloneqq (x_1 - \alpha_1, \ldots, x_n - \alpha_n).
            \]
        \item
            The map $\mathcal{G}_{\text{internal}} : \F^{n+s} \to \F^s$ is defined as follows.
            Let $V$ be the set of gates of $\Phi$.
            Let $v_1, \ldots, v_s$ be a topological ordering of the internal gates of $\Phi$, where $v_s$ is the output gate of $\Phi$.
            Define the function $L : V \to \F \cup \set{x_1, \ldots, x_n, y_1, \ldots, y_s}$ via
            \[
                L(v) \coloneqq
                \begin{cases}
                    \gamma & \text{if $v$ is an input gate labeled by a constant $\gamma \in \F$,} \\
                    x_i & \text{if $v$ is an input gate labeled by the variable $x_i$,} \\
                    y_i & \text{if $v$ is the $i$\ts{th} internal gate in topological order.}
                \end{cases}
            \]
            The $i$\ts{th} output of $\mathcal{G}_{\mathrm{internal}}$ is defined in terms of the $i$\ts{th} internal gate $v_i$ and its children $u$ and $w$.
            \begin{itemize}
                \item 
                    If $v_i$ is an addition gate, then the $i$\ts{th} output of $\mathcal{G}_{\mathrm{internal}}$ is the polynomial $L(v_i) - (L(u) + L(w))$.
                \item 
                    If $v_i$ is a product gate, then the $i$\ts{th} output of $\mathcal{G}_{\mathrm{internal}}$ is the polynomial $L(v_i) - L(u) L(w)$.
            \end{itemize}
        \item
            The polynomial map $\mathcal{G}_{\text{output}} : \F^{n+s} \to \F$ is given by
            \[
                \mathcal{G}_{\mathrm{output}}(\vec{x}, \vec{y}) \coloneqq y_s - \beta. \qedhere
            \]
    \end{enumerate}
\end{definition}

The name ``local encoding'' arises from the fact that the local encoding $\mathcal{G}$ of $\Phi(\vec{\alpha}) = \beta$ corresponds to a system of polynomial equations that is satisfiable if and only if the circuit $\Phi$ outputs $\beta$ when evaluated at $\vec{\alpha}$.
The corresponding system of polynomial equations is $\mathcal{G}(\vec{x}, \vec{y}) = (0, \ldots, 0)$.
The first $n$ equations $\mathcal{G}_{\text{input}}(\vec{x}, \vec{y}) = (0, \ldots, 0)$ ensure that the input to the circuit $\Phi$ is the point $\vec{\alpha}$.
The next $s$ equations $\mathcal{G}_{\text{internal}}(\vec{x}, \vec{y}) = (0, \ldots, 0)$ enforce that for each $i \in [s]$, the value of the variable $y_i$ equals the value of the $i$\ts{th} internal gate when $\Phi$ is evaluated at $\vec{\alpha}$.
Finally, the equation $\mathcal{G}_{\text{output}}(\vec{x}, \vec{y}) = 0$ expresses that the output gate of $\Phi$ evaluates to $\beta$.

The following lemma describes the parameters of a local encoding, which follow immediately from \cref{def:local encoding}.

\begin{lemma} \label{lem:generator parameters}
    Let $\Phi$ be an $n$-variate arithmetic circuit of size $s$ and fan-in two, and let $\vec{\alpha} \in \F^n$ and $\beta \in \F$.
    Let $\mathcal{G} : \F^{n+s} \to \F^{n+s+1}$ be the local encoding of $\Phi(\vec{\alpha}) = \beta$.
    Then the following hold.
    \begin{enumerate}
        \item 
            The seed length of $\mathcal{G}$ is $n+s$ and the stretch of $\mathcal{G}$ is $1$.
        \item 
            The degree of $\mathcal{G}$ is $2$.
        \item 
            Every output of $\mathcal{G}$ can be computed by an arithmetic formula of size $2$.
    \end{enumerate}
\end{lemma}

In the subsections to come, we determine the annihilator ideals of local encodings and use this to infer pseudorandom properties of local encodings when the circuit $\Phi$ computes a sufficiently-hard polynomial.
Before moving on to these general results, we first work out a small, concrete example to build intuition for how annihilators of local encodings behave.

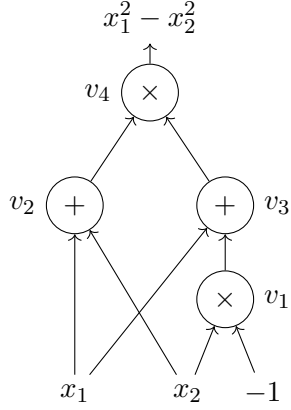
\begin{figure}
    \begin{center}
        \begin{tikzpicture}
            \tikzstyle{gate}=[circle,draw=black]

            \node (i1) at (-1, 0) {$x_1$};
            \node (i2) at (0.5, 0) {$x_2$};
            \node (i3) at (1.5, 0) {$-1$};

            \node[gate, label=right:$v_1$] (v1) at (1, 1.25) {$\times$};
            \node[gate, label=left:$v_2$] (v2) at (-1, 2.5) {$+$};
            \node[gate, label=right:$v_3$] (v3) at (1, 2.5) {$+$};
            \node[gate, label=left:$v_4$] (v4) at (0, 4) {$\times$};
            \node (output) at (0,5) {$x_1^2 - x_2^2$};
            
            \path[->] (i1) 
                edge (v2) 
                edge (v3);
            \path[->] (i2)
                edge (v1)
                edge (v2);
            \path[->] (i3) edge (v1);
            \path[->] (v2) edge (v4);
            \path[->] (v1) edge (v3);
            \path[->] (v3) edge (v4);
            \path[->] (v4) edge (output);
        \end{tikzpicture}
    \end{center}
    \caption{An arithmetic circuit computing the polynomial $x_1^2 - x_2^2$.}
    \label{fig:local encoding example}
\end{figure}

\begin{example} \label{ex:local encoding}
    Let $\Phi$ be the arithmetic circuit depicted in \Cref{fig:local encoding example} that computes the polynomial $x_1^2 - x_2^2$.
    The internal gates of $\Phi$ are labeled as $v_1$, $v_2$, $v_3$, and $v_4$, corresponding to a topological ordering of the internal gates.
    For $\alpha_1, \alpha_2, \beta \in \F$, the local encoding of $\Phi(\alpha_1, \alpha_2) = \beta$ is the polynomial map $\mathcal{G} : \F^6 \to \F^7$ given by
    \[
        \mathcal{G}(x_1,x_2,y_1,y_2,y_3,y_4) = (x_1 - \alpha_1, x_2 - \alpha_2, y_1 + x_2, y_2 - x_1 - x_2, y_3 - x_1 - y_1, y_4 - y_2 y_3, y_4 - \beta).
    \]

    What do the annihilators of $\mathcal{G}$ look like?
    One annihilator can be constructed by writing $y_4$ as a polynomial combination of the first $6$ outputs of $\mathcal{G}$ and then taking the difference of this polynomial and $z_7+\beta$.
    To do this, we iteratively find polynomials $h_i \in \F[z_1,\ldots,z_7]$ so that $(h_i \circ \mathcal{G})(\vec{x}, \vec{y}) = y_i$.
    The desired polynomials $h_1$, $h_2$, $h_3$, and $h_4$ are given by
    \begin{align*}
        h_1(\vec{z}) &= z_3 - (z_2 + \alpha_2) \\
        h_2(\vec{z}) &= z_4 + (z_1 + \alpha_1) + (z_2 + \alpha_2) \\
        h_3(\vec{z}) &= z_5 + (z_1 + \alpha_1) + h_1(\vec{z}) \\
        h_4(\vec{z}) &= z_6 + h_2(\vec{z}) h_3(\vec{z}).
    \end{align*}
    Because $(h_4 \circ \mathcal{G})(\vec{x}, \vec{y}) = y_4$, it follows that 
    \[
        h(\vec{z}) \coloneqq h_4(\vec{z}) - (z_7 + \beta)
    \]
    is a nonzero annihilator of $\mathcal{G}$.
    Because $\Ann{\mathcal{G}}$ is an ideal, every multiple of $h$ is also an annihilator of $\mathcal{G}$.
    Using computer software, such as Macaulay2, one can verify that these are the only annihilators of $\mathcal{G}$: the ideal $\Ann{\mathcal{G}}$ is precisely the principal ideal generated by $h$.

    To prove that local encodings are pseudorandom, we need to understand the complexity of their annihilators.
    As a first step towards this, let's understand how the complexity of $h$ relates to the complexity of $x_1^2 - x_2^2$, the polynomial computed by $\Phi$.
    Expanding out $h(\vec{z})$ as
    \begin{multline*}
        h(\vec{z}) = z_1^2 - z_2^2 + z_1 z_3 + z_2 z_3 + z_1 z_4 - z_2 z_4 + z_3 z_4 + z_1 z_5 + z_2 z_5 + z_4 z_5 + 2 \alpha_1 z_1 - 2 \alpha_2 z_2 \\ + (\alpha_1 + \alpha_2)z_3 + (\alpha_1 - \alpha_2) z_4 + (\alpha_1 + \alpha_2) z_5 + z_6 - z_7 + \alpha_1^2 - \alpha_2^2 - \beta,
    \end{multline*}
    we rewrite $h(\vec{z})$ as
    \[
        h(\vec{z}) = \del{(z_1+\alpha_1)^2 - (z_2 + \alpha_2)^2} + g(\vec{z}) - z_7 - \beta,
    \]
    where
    \begin{multline*}
        g(\vec{z}) \coloneqq z_1 z_3 + z_2 z_3 + z_1 z_4 - z_2 z_4 + z_3 z_4 + z_1 z_5 + z_2 z_5 + z_4 z_5 \\ + (\alpha_1 + \alpha_2)z_3 + (\alpha_1 - \alpha_2) z_4 + (\alpha_1 + \alpha_2) z_5 + z_6.
    \end{multline*}
    Importantly, the polynomial $g$ is an element of the ideal $\abr{z_3,z_4,z_5,z_6} \subseteq \F[\vec{z}]$.
    In other words, every monomial of $g$ is divisible by one of the variables $z_3$, $z_4$, $z_5$, or $z_6$.
    By setting $z_3 = \cdots = z_7 = 0$, we see that $h$ projects to $\del{(z_1 + \alpha)^2 - (z_2 + \alpha_2)^2} - \beta$.
    Applying the change of variables $(z_1, z_2) \mapsto (z_1 - \alpha_1, z_2 - \alpha_2)$, this projection of $h$ becomes $\del{z_1^2 - z_2^2} - \beta$.
    Adding $\beta$ yields $z_1^2 - z_2^2$, the polynomial computed by the circuit $\Phi$.

    By zeroing out some variables, shifting other variables by a constant, and adding an appropriate constant, the generator $h$ is transformed to $z_1^2 - z_2^2$, the polynomial that the circuit $\Phi$ computes.
    This transformation of $h$ has low complexity, so a small circuit that computes $h$ can be used to compute $z_1^2 - z_2^2$ with similar complexity.
    In the contrapositive, a lower bound on the complexity of $z_1^2 - z_2^2$ implies a comparable lower bound on $h$.
    This reduction is not particularly useful for the specific example at hand, as the polynomial $z_1^2-z_2^2$ is easy to compute.

    This example illustrates some features of the general case.
    The ideal $\Ann{\mathcal{G}}$ is always principal (\cref{lem:ann ideal principal}) and is generated by a polynomial with the same structure as $h$ above (\cref{prop:ann ideal generator}).
    Using this structure, we can prove lower bounds on the complexity of $h$ and its multiples by appealing to lower bounds on the complexity of the polynomial computed by $\Phi$ and its multiples (\cref{lem:annihilator hardness}).
\end{example}

\subsection{Annihilators of Local Encodings} \label{subsec:annihilator ideal}

Having defined our candidate generator $\mathcal{G}$ as the local encoding of $\Phi(\vec{\alpha}) = \beta$ for an arithmetic circuit $\Phi$, we now proceed to analyze its annihilator ideal $\Ann{\mathcal{G}}$.
Our ultimate goal is to prove lower bounds on the complexity of every nonzero polynomial in $\Ann{\mathcal{G}}$, as this equates (via \cref{lem:hitting vs annihilators}) to proving that $\mathcal{G}$ hits some circuit class.
To do this, we first need to understand what the polynomials in $\Ann{\mathcal{G}}$ look like.

We begin by showing that $\Ann{\mathcal{G}}$ is a nonzero principal ideal.
This means that there is a single polynomial $h(\vec{z})$ such that every element of $\Ann{\mathcal{G}}$ is a multiple of $h$.
This makes the task of proving lower bounds for annihilators much easier, as we only have to reason about multiples of a single polynomial.

The following lemma is essentially due to \textcite[Lemma 7]{Kayal09}.
We cannot directly apply \cite[Lemma 7]{Kayal09}, as the polynomials in a local encoding do not necessarily fit the hypothesis of the lemma as it appears in Kayal's work.
However, inspecting the proof, it is clear that Kayal's argument applies to our setting.
We include a (nearly identical) proof for the sake of completeness.

\begin{lemma} \label{lem:ann ideal principal}
    Let $\Phi$ be an $n$-variate arithmetic circuit of size $s$ and let $\mathcal{G}$ be the local encoding of $\Phi(\vec{\alpha}) = \beta$.
    Then the ideal $\Ann{\mathcal{G}}$ is a nonzero principal ideal.
\end{lemma}

\begin{proof}
    The polynomial map $\mathcal{G}$ consists of $n+s+1$ polynomials in $n+s$ variables.
    Because there are more polynomials than there are variables, the outputs of $\mathcal{G}$ are algebraically dependent, so the annihilator ideal $\Ann{\mathcal{G}}$ is nonempty.

    Let $a(z_1,\ldots,z_{n+s+1}) \in \Ann{\mathcal{G}}$ be a nonzero element of minimal degree.
    We claim that $a(\vec{z})$ generates the ideal $\Ann{\mathcal{G}}$.
    To prove this, it suffices to show that for any nonzero $b(\vec{z}) \in \Ann{\mathcal{G}}$, the polynomial $a(\vec{z})$ divides $b(\vec{z})$.

    We first show that $a(\vec{z})$ is irreducible.
    Suppose $a(\vec{z})$ factors as $a(\vec{z}) = a_1(\vec{z}) a_2(\vec{z})$.
    Then we have
    \[
        0 = a(\mathcal{G}(\vec{x}, \vec{y})) = a_1(\mathcal{G}(\vec{x}, \vec{y})) \cdot a_2(\mathcal{G}(\vec{x},\vec{y})),
    \]
    so either $a_1(\mathcal{G}(\vec{x},\vec{y})) = 0$ or $a_2(\mathcal{G}(\vec{x},\vec{y})) = 0$.
    Suppose, without loss of generality, that $a_1(\mathcal{G}(\vec{x},\vec{y})) = 0$.
    Then by definition we have $a_1(\vec{z}) \in \Ann{\mathcal{G}}$.
    Because $a$ was chosen to be a nonzero element of $\Ann{\mathcal{G}}$ of minimal degree, we have $\deg(a) \le \deg(a_1)$.
    On the other hand, because $a_1$ divides $a$, we have $\deg(a_1) \le \deg(a)$.
    Together, these inequalities imply $\deg(a) = \deg(a_1)$.
    It follows that $\deg(a_2) = 0$, so $a_2(\vec{z})$ is a nonzero constant polynomial.
    Hence $a(\vec{z})$ is irreducible as claimed.

    Now let $b(\vec{z}) \in \Ann{\mathcal{G}}$ be nonzero.
    To see that $a(\vec{z})$ divides $b(\vec{z})$, consider their resultant with respect to $z_{n+s+1}$,
    \[
        r(\vec{z}) \coloneqq \res_{z_{n+s+1}}(a(\vec{z}), b(\vec{z})).
    \]
    By \cref{lem:resultant elimination}, we have
    \[
        r(\vec{z}) \in \abr{a(\vec{z}), b(\vec{z})} \cap \F[z_1, \ldots, z_{n+s}] \subseteq \Ann{\mathcal{G}} \cap \F[z_1,\ldots,z_{n+s}].
    \]
    This implies $r(\vec{z})$ vanishes on the first $n+s$ outputs of $\mathcal{G}$.
    If $r(\vec{z})$ were nonzero, then this vanishing would imply that the first $n+s$ outputs of $\mathcal{G}$ are algebraically dependent.
    However, this is impossible: the Jacobian of the first $n+s$ outputs of $\mathcal{G}$ is a triangular matrix with ones along the diagonal, so \cref{prop:jacobian arbitrary field} implies they are algebraically independent.
    Thus, it follows that $r(\vec{z}) = 0$.
    \cref{lem:resultant common factor multivariate} implies that $a(\vec{z})$ and $b(\vec{z})$ have a common factor.
    Because $a(\vec{z})$ is irreducible, it follows that $a(\vec{z})$ divides $b(\vec{z})$ as desired.
\end{proof}

Our next goal is to describe a polynomial that generates the ideal $\Ann{\mathcal{G}}$.
Just as in \cref{ex:local encoding}, we can find a dependency between the outputs of $\mathcal{G}$ by iteratively expressing each of the $y$ variables as a polynomial combination of the first $n+s$ outputs of $\mathcal{G}$.
Doing this results in a polynomial $h_s(\vec{z})$ that satisfies $h_s(\mathcal{G}(\vec{x}, \vec{y})) = y_s$.
Subtracting $z_{n+s+1}+\beta$ from $h_s(\vec{z})$ results in an annihilator of $\mathcal{G}$.
It turns out that the annihilator obtained in this way is irreducible, and hence it generates the ideal $\Ann{\mathcal{G}}$.

We start with the following lemma, which shows that each of the variables $y_1, \ldots, y_s$ can be obtained as a polynomial combination of the first $n+s$ outputs of $\mathcal{G}$.
These polynomials are obtained by simulating the underlying circuit $\Phi$ in a gate-by-gate manner.
As a byproduct of this construction, we also obtain a bound on the circuit complexity of these polynomials.

\begin{lemma} \label{lem:annihilator inductive step}
    Let $\Phi$ be an $n$-variate arithmetic circuit of size $s$ and let $\mathcal{G}$ be the local encoding of $\Phi(\vec{\alpha}) = \beta$.
    For every $k \in [s]$, there is a multi-output arithmetic circuit $\Psi_k$ such that the following hold.
    \begin{itemize}
        \item 
            The size of $\Psi_k$ is bounded by $O(k)$.
        \item
            The circuit $\Psi_k$ outputs $k$ polynomials $h_1, \ldots, h_k \in \F[z_1,\ldots,z_{n+s+1}]$ that depend only on the variables $z_1, \ldots, z_{n+k}$.
            For each $i \in [k]$, the polynomial $h_i$ satisfies the identity $h_i(\mathcal{G}(\vec{x}, \vec{y})) = y_i$.
    \end{itemize}
\end{lemma}

\begin{proof}
    We proceed by induction on $k$.
    Let $\Psi_{k-1}$ be the circuit given by induction.
    (If $k = 1$, then $\Psi_{0}$ is the empty circuit.)
    To construct $\Psi_k$, it suffices to add a constant number of new gates to $\Psi_{k-1}$ in order to compute the desired polynomial $h_k(\vec{z})$.

    To avoid a multitude of sub-cases, we need some additional notation.
    Let $V$ be the set of gates of $\Phi$.
    Recall the function $L : V \to \F \cup \set{x_1,\ldots,x_n, y_1,\ldots,y_s}$ used in the definition of $\mathcal{G}$, which was given by
    \[
        L(v) \coloneqq
        \begin{cases}
            \gamma & \text{if $v$ is an input gate labeled by a constant $\gamma \in \F$,} \\
            x_i & \text{if $v$ is an input gate labeled by the variable $x_i$,} \\
            y_i & \text{if $v$ is the $i$\ts{th} internal gate in topological order.}
        \end{cases}
    \]
    Let $V_{\le {k-1}}$ be the set consisting of the input gates and first $k-1$ internal gates of $\Phi$.
    Define the function $\hat{L} : V_{\le k-1} \to \F[\vec{z}]$ via
    \[
        \hat{L}(v) \coloneqq \begin{cases}
            \gamma & \text{if $v$ is an input gate labeled by a constant $\gamma \in \F$,} \\
            z_i + \alpha_i & \text{if $v$ is an input gate labeled by the variable $x_i$,} \\
            h_i(\vec{z}) & \text{if $v$ is the $i$\ts{th} internal gate in topological order.}
        \end{cases}
    \]
    Note that every output of $\hat{L}$ is either already computed by the circuit $\Psi_{k-1}$ or can be computed using a constant number of gates.
    An immediate consequence of the inductive hypothesis is that the polynomial $\hat{L}(v) \circ \mathcal{G}(\vec{x},\vec{y})$ is given by
    \[
        \hat{L}(v) \circ \mathcal{G}(\vec{x},\vec{y}) = \begin{cases}
            \gamma & \text{if $v$ is an input gate labeled by a constant $\gamma \in \F$,} \\
            x_i & \text{if $v$ is an input gate labeled by the variable $x_i$,} \\
            y_i & \text{if $v$ is the $i$\ts{th} internal gate in topological order.}
        \end{cases}
    \]
    Put more succinctly, we have the identity
    \[
        \hat{L}(v) \circ \mathcal{G}(\vec{x}, \vec{y}) = L(v)
    \]
    for all gates $v \in V_{\le k-1}$.

    We now describe how to compute the desired polynomial $h_k(\vec{z})$.
    Let $v$ be the $k$\ts{th} internal gate of $\Phi$, and let $u$ and $w$ be its children.
    There are two cases to consider, depending on whether $v$ is an addition or multiplication gate.
    
    \begin{itemize}
        \item
            Suppose $v = u + w$.
            Define 
            \[
                h_k(\vec{z}) \coloneqq z_{n+k} + \hat{L}(u) + \hat{L}(w).
            \]
            It is clear that $h_k(\vec{z})$ can be computed by adding $O(1)$ gates to $\Psi_{k-1}$, so the size of $\Psi_k$ is bounded by $O(k)$ as claimed.

            It remains to verify that $h_k$ satisfies the claimed identity.
            Composing with $\mathcal{G}(\vec{x}, \vec{y})$, we have
            \begin{align*}
                h_k(\mathcal{G}(\vec{x}, \vec{y})) &= y_k - L(u) - L(w) + \hat{L}(u) \circ \mathcal{G}(\vec{x}, \vec{y}) + \hat{L}(w) \circ \mathcal{G}(\vec{x}, \vec{y}) \\
                &= y_k - L(u) - L(w) + L(u) + L(w) \\
                &= y_k
            \end{align*}
            as desired.
        \item
            Suppose $v = u \times w$.
            Define
            \[
                h_k(\vec{z}) \coloneqq z_{n+k} + \hat{L}(u) \hat{L}(w).
            \]
            As in the previous case, it is clear that $h_k(\vec{z})$ can be computed by adding $O(1)$ gates to $\Psi_{k-1}$, so the size of $\Psi_k$ is bounded by $O(k)$.

            We now verify the claimed identity on $h_k$.
            We compose with $\mathcal{G}(\vec{x}, \vec{y})$ to obtain
            \begin{align*}
                h_k(\mathcal{G}(\vec{x}, \vec{y})) &= y_k - L(u) L(w) + (\hat{L}(u) \circ \mathcal{G}(\vec{x}, \vec{y})) \cdot (\hat{L}(w) \circ \mathcal{G}(\vec{x}, \vec{y})) \\
                &= y_{k} - L(u) L(w) + L(u) L(w) \\
                &= y_k 
            \end{align*}
            as desired.
    \end{itemize}
    In both cases, we can construct $\Psi_k$ by adding $O(1)$ gates to $\Psi_{k-1}$.
\end{proof}

Using the preceding lemma, we now describe the generator of the ideal $\Ann{\mathcal{G}}$.

\begin{proposition} \label{prop:ann ideal generator}
    Let $\Phi$ be an $n$-variate arithmetic circuit of size $s$ that computes a polynomial $f(x_1,\ldots,x_n)$.
    Let $\mathcal{G}$ be the local encoding of $\Phi(\vec{\alpha}) = \beta$.
    Then there is a polynomial $h \in \F[z_1,\ldots,z_{n+s+1}]$ such that the following hold.
    \begin{enumerate}
        \item 
            The ideal $\Ann{\mathcal{G}}$ is generated by $h(\vec{z})$.
        \item
            There exists a polynomial $g \in \F[z_1,\ldots,z_{n+s}]$, not depending on the variable $z_{n+s+1}$, such that $g \in \abr{z_{n+1}, \ldots, z_{n+s}}$ and 
            \[
                h(\vec{z}) = z_{n+s+1} - f(z_1 + \alpha_1, \ldots, z_n + \alpha_n) + g(\vec{z}) + \beta.
            \]
        \item
            There is an arithmetic circuit of size $O(s)$ that computes $h(\vec{z})$.
    \end{enumerate}
\end{proposition}

\begin{proof}
    Let $\Psi_s$ be the circuit of size $O(s)$ obtained by applying \cref{lem:annihilator inductive step} to $\Phi$ and $\mathcal{G}$.
    Let $h_s(\vec{z})$ be the output of $\Psi_s$ that satisfies $h_s(\mathcal{G}(\vec{x}, \vec{y})) = y_s $.
    Define
    \[
        h(\vec{z}) \coloneqq z_{n+s+1} - h_s(\vec{z}) + \beta.
    \]
    It is clear from this definition that $h(\vec{z})$ can be computed by an arithmetic circuit of size $O(s)$.
    It remains to verify that $h$ generates $\Ann{\mathcal{G}}$ and that $h$ has the claimed form.

    We first establish that $h$ generates the ideal $\Ann{\mathcal{G}}$.
    To see that $h \in \Ann{\mathcal{G}}$, we compute
    \begin{align*}
        h(\mathcal{G}(\vec{x}, \vec{y})) &= y_s - \beta - h_s(\mathcal{G}(\vec{x}, \vec{y})) + \beta \\
        &= y_s - \beta - y_s + \beta \\
        &= 0.
    \end{align*}
    Thus $h \in \Ann{\mathcal{G}}$.
    To see that $h$ generates $\Ann{\mathcal{G}}$, note that because $h_s(\vec{z})$ does not depend on the variable $z_{n+s+1}$, the polynomial $h$ is monic and has degree $1$ in $z_{n+s+1}$.
    This implies that $h$ is irreducible.
    Because $h$ is irreducible, the ideal $\Ann{\mathcal{G}}$ is principal, and $h \in \Ann{\mathcal{G}}$, it follows that $h$ generates $\Ann{\mathcal{G}}$.

    It remains to show that $h$ has the claimed form.
    Recall that $\mathcal{G}$ is the concatenation of the three polynomial maps $\mathcal{G}_{\textrm{input}}$, $\mathcal{G}_{\textrm{internal}}$, and $\mathcal{G}_{\textrm{output}}$.
    Because $h \in \Ann{\mathcal{G}}$, we have the identity
    \[
        h(\mathcal{G}_{\textrm{input}}(\vec{x}, \vec{y}), \mathcal{G}_{\textrm{internal}}(\vec{x}, \vec{y}), \mathcal{G}_{\textrm{output}}(\vec{x}, \vec{y})) = 0.
    \]
    Let $f_i(\vec{x})$ be the polynomial computed by the $i$\ts{th} internal gate of $\Phi$ and consider the substitution given by $y_i \mapsto f_i(\vec{x})$.
    We claim that under this substitution, the outputs of $\mathcal{G}_{\text{internal}}$ are mapped to zero.
    Indeed, if the $i$\ts{th} gate of $\Phi$ is an addition gate whose children are the $j$\ts{th} and $k$\ts{th} internal gates, then under this substitution the $i$\ts{th} output of $\mathcal{G}_{\textrm{internal}}$ becomes
    \[
        f_i(\vec{x}) - f_j(\vec{x}) - f_k(\vec{x}) = f_i(\vec{x}) - f_i(\vec{x}) = 0.
    \]
    Likewise, if the $i$\ts{th} internal gate of $\Phi$ is a multiplication gate whose children are the $j$\ts{th} and $k$\ts{th} internal gates, then the $i$\ts{th} output of $\mathcal{G}_{\text{internal}}$ is mapped to
    \[
        f_i(\vec{x}) - f_j(\vec{x}) f_k(\vec{x}) = f_i(\vec{x}) - f_i(\vec{x}) = 0.
    \]
    Similar cancellations occur when the $i$\ts{th} internal gate has one or more input gates as children; we omit the straightforward calculations.

    Applying the substitution $y_i \mapsto f_i(\vec{x})$ to the equation $h(\mathcal{G}(\vec{x}, \vec{y})) = 0$, we obtain the identity
    \[
        h(x_1 - \alpha_1, \ldots, x_n - \alpha_n, 0, \ldots, 0, f(\vec{x}) - \beta) = 0.
    \]
    Next, we apply the shift $x_i \mapsto x_i + \alpha_i$ to obtain
    \[
        h(x_1, \ldots, x_n, 0, \ldots, 0, f(\vec{x} + \vec{\alpha}) - \beta) = 0.
    \]
    This implies that $z_{n+s+1} - f(\vec{x} + \vec{\alpha}) + \beta$ is a factor of $h(x_1, \ldots, x_n, 0, \ldots, 0, z_{n+s+1})$.
    Because the polynomial $h(x_1, \ldots, x_n, 0, \ldots, 0, z_{n+s+1})$ is monic and degree $1$ in $z_{n+s+1}$, we obtain the equality
    \[
        h(x_1, \ldots, x_n, 0, \ldots, 0, z_{n+s+1}) = z_{n+s+1} - f(\vec{x} + \vec{\alpha}) + \beta.
    \]
    This implies
    \[
        h(\vec{z}) = z_{n+s+1} - f(\vec{z} + \vec{\alpha}) + g(\vec{z}) + \beta,
    \]
    where $g \in \F[\vec{z}]$ is a polynomial that vanishes on the substitution $z_{n+1} = \cdots = z_{n+s} = 0$, which is equivalent to $g$ being an element of the ideal $\abr{z_{n+1}, \ldots, z_{n+s}}$.
    Thus $h(\vec{z})$ has the claimed form.
\end{proof}

\subsection{Pseudorandomness from Hardness of Multiples}

In this subsection, we investigate the pseudorandom properties of local encodings.
Suppose $\mathcal{G}$ is the local encoding of $\Phi(\vec{\alpha}) = \beta$, where $\Phi$ is an arithmetic circuit that computes a polynomial $f(\vec{x})$.
As we saw in \cref{prop:ann ideal generator}, the ideal $\Ann{\mathcal{G}}$ is generated by a polynomial $h(\vec{z})$ that is similar to $f(\vec{x})$, in the sense that any circuit computing $h$ can be modified in a simple way to produce a circuit that computes $f$.
This reduction shows that lower bounds on the complexity of $f$ imply lower bounds on the complexity of $h$.

To show that $\mathcal{G}$ hits a circuit class $\mathscr{C}$, we need to show that every nonzero multiple of $h$ is hard for $\mathscr{C}$, not just that $h$ itself is hard for $\mathscr{C}$.
Because the polynomials $h$ and $f$ are closely related, it seems reasonable to expect that lower bounds on the complexity of multiples of $f$ imply comparable lower bounds for multiples of $h$.

This intuition is correct: suppose $p \in \Ann{\mathcal{G}}$ is a polynomial of the form $h(\vec{z}) \cdot r(\vec{z})$.
If we work in the ring $\F[z_1,\ldots,z_n][z_{n+1},\ldots,z_{n+s+1}]$, then the bottom-degree homogeneous component of $h(\vec{z})$ is precisely
\[
    h_{\text{bot}}(\vec{z}) = -f(z_1 + \alpha_1, \ldots, z_n + \alpha_n) + \beta.
\]
Bottom-degree homogeneous components are homomorphic with respect to multiplication, so the bottom-degree component of $h(\vec{z}) \cdot r(\vec{z})$ is a multiple of $h_{\text{bot}}(\vec{z})$.
From a circuit computing $h(\vec{z}) \cdot r(\vec{z})$, we can obtain circuit for this multiple of $h_{\text{bot}}(\vec{z})$ by applying polynomial interpolation in the variables $z_{n+1},\ldots,z_{n+s+1}$.
Applying the shift $z_i \mapsto z_i-\alpha_i$ for $i \in [n]$ results in a circuit that computes a multiple of $f(z_1,\ldots,z_n) - \beta$.
Thus, if multiples of $f - \beta$ require large circuits, then it must have been the case that the circuit computing $h(\vec{z}) \cdot r(\vec{z})$ was large to begin with, so we obtain a lower bound for all nonzero polynomials in $\Ann{\mathcal{G}}$.

We now make the preceding sketch precise, showing that lower bounds on $\Ann{\mathcal{G}}$ follow from lower bounds for multiples of $f(\vec{x})$.
Although the interpolation argument described above works in general, we state and prove this lemma for three concrete measures of complexity (bounded-depth circuits, formulas, and branching programs) that correspond to the applications we give in \cref{sec:instantiate}.

\begin{lemma} \label{lem:annihilator hardness} 
    Let $\F$ be an infinite field.
    Let $\Phi$ be an $n$-variate circuit arithmetic circuit that computes a polynomial $f \in \F[\vec{x}]$.
    Let $\mathcal{G}$ be the local encoding of $\Phi(\vec{\alpha}) = \beta$.
    \begin{enumerate}
        \item 
            If every nonzero multiple of $f(\vec{x}) - \beta$ requires depth-$\Delta$ circuits of size at least $s$, then every nonzero polynomial $p \in \Ann{\mathcal{G}}$ requires depth-$(\Delta - 2)$ circuits of size at least $\Omega(s^{\frac{1}{\Delta-1}})$.
        \item
            If every nonzero multiple of $f(\vec{x}) - \beta$ requires formulas of size at least $s$, then every nonzero polynomial $p \in \Ann{\mathcal{G}}$ requires formulas of size at least $\Omega(\sqrt{s})$.
        \item
            If every nonzero multiple of $f(\vec{x}) - \beta$ requires branching programs of size at least $s$, then every nonzero polynomial $p \in \Ann{\mathcal{G}}$ requires branching programs of size at least $\Omega(\sqrt{s})$.
    \end{enumerate}
\end{lemma}

\begin{proof}
    Let $p \in \Ann{\mathcal{G}}$ be nonzero.
    By \cref{prop:ann ideal generator}, we know that $p(\vec{z})$ can be written as
    \[
        p(\vec{z}) = \del{z_{n+s+1} - f(z_1 + \alpha_1, \ldots, z_n + \alpha_n) + g(\vec{z}) + \beta} \cdot r(\vec{z}),
    \]
    where $g \in \abr{z_{n+1}, \ldots, z_{n+s}}$ and $r(\vec{z})$ is a nonzero polynomial.
    Consider the substitution
    \[
        z_i \mapsto 
        \begin{cases}
            z_i - \alpha_i & \text{if $1 \le i \le n$,} \\
            w \cdot z_i & \text{if $n+1 \le i \le n+s+1$,}
        \end{cases}
    \]
    where $w$ is a fresh variable.
    Under this substitution, the first factor of $p$ is mapped to
    \[
        - f(z_1, \ldots, z_n) + \beta + O(w),
    \]
    where $O(w)$ indicates a term divisible by $w$.
    Let
    \[
        \hat{r}(\vec{z}, w) = \sum_{i=a}^b r_i(\vec{z}) w^i
    \]
    be the image of $r(\vec{z})$ under this substitution, where $a \le b$ and $r_a(\vec{z}) \neq 0$.
    Then the image of $p(\vec{z})$ under this substitution is
    \[
        \hat{p}(\vec{z}, w) \coloneqq (-f(\vec{z}) + \beta) \cdot r_a(\vec{z}) w^a + O(w^{a+1}),
    \]
    where $O(w^{a+1})$ indicates a term divisible by $w^{a+1}$.

    We now obtain a circuit that computes $(-f(\vec{z}) + \beta) \cdot r_a(\vec{z})$ using polynomial interpolation with respect to $w$.
    Note that $\hat{p}(\vec{z}, w)$ has degree at most $d \coloneqq \deg(p)$ with respect to $w$.
    Let $\gamma_0, \ldots, \gamma_{d} \in \F$ be distinct field elements.
    Then there are constants $\zeta_0, \ldots, \zeta_{d} \in \F$ such that
    \[
        (-f(\vec{z}) + \beta) \cdot r_a(\vec{z}) = \sum_{i=0}^d \zeta_i \cdot \hat{p}(\vec{z}, \gamma_i).
    \]
    This expresses a nonzero multiple of $f(\vec{z}) - \beta$ as a sum of $\deg(p)+1$ copies of $\hat{p}(\vec{z}, w)$, each of which can be obtained in a simple manner from the circuit computing $p(\vec{z})$.
    In particular, if nonzero multiples of $f(\vec{z}) - \beta$ have high complexity, then a similar lower bound holds for the annihilator $p \in \Ann{\mathcal{G}}$ that we started with.
    The precise details of this bound depends on the complexity measure of interest.
    \begin{enumerate}
        \item 
            Suppose $p$ was computed by a circuit of depth $\Delta-2$ and size $t$.
            Then the resulting expression for $(-f(\vec{z}) + \beta) \cdot r_a(\vec{z})$ can be implemented by a circuit of depth $\Delta$ and size $O(t \deg(p))$.
            Because $p$ was computed by a circuit of depth $\Delta - 2$ and size $t$, we have the bound $\deg(p) \le t^{\Delta - 2}$.
            Thus, we have a circuit of depth $\Delta$ and size $O(t^{\Delta-1})$ that computes $(-f(\vec{z}) + \beta) \cdot r_a(\vec{z})$.
            If every nonzero multiple of $f(\vec{z}) - \beta$ requires depth-$\Delta$ circuits of size $s$, we conclude $t \ge \Omega(s^{\frac{1}{\Delta-1}})$ as desired.
        \item
            Suppose $p$ was computed by a formula of size $t$.
            As in the previous case, the expression for $(-f(\vec{z}) + \beta) \cdot r_a(\vec{z})$ can be implemented by a formula of size $O(t \deg(p))$.
            Because $p$ was computed by a formula of size $t$, we have the bound $\deg(p) \le t$.
            Thus, we obtain a formula of size $O(t^2)$ that computes a nonzero multiple of $f(\vec{z}) - \beta$.
            If such polynomials require formulas of size $s$, we conclude the desired lower bound of $t \ge \Omega(\sqrt{s})$.
        \item
            Suppose $p$ was computed by a branching program of size $t$.
            The argument in this case is identical to the preceding argument for formulas, so we again conclude the lower bound $t \ge \Omega(\sqrt{s})$.
            \qedhere
    \end{enumerate}
\end{proof}

As an immediate corollary of the preceding lemma, we see that $\mathcal{G}$ hits circuit classes that require large size to compute multiples of $f(\vec{x}) - \beta$.

\begin{corollary} \label{cor:hardness to randomness}
    Let $\F$ be an infinite field.
    Let $\Phi$ be an $n$-variate arithmetic circuit that computes a polynomial $f \in \F[\vec{x}]$.
    Let $\mathcal{G}$ be the local encoding of $\Phi(\vec{\alpha}) = \beta$.
    \begin{enumerate}
        \item 
            Suppose every nonzero multiple of $f(\vec{x}) - \beta$ requires depth-$\Delta$ circuits of size at least $s$.
            Then $\mathcal{G}$ hits circuits of depth $\Delta - 2$ and size $\eps s^{\frac{1}{\Delta-1}}$ for a sufficiently small constant $\eps > 0$.
        \item
            Suppose every nonzero multiple of $f(\vec{x}) - \beta$ requires formulas of size at least $s$.
            Then $\mathcal{G}$ hits formulas of size $\eps \sqrt{s}$ for a sufficiently small constant $\eps > 0$.
        \item
            Suppose every nonzero multiple of $f(\vec{x}) - \beta$ requires branching programs of size at least $s$.
            Then $\mathcal{G}$ hits branching programs of size $\eps \sqrt{s}$ for a sufficiently small constant $\eps > 0$.
    \end{enumerate}
\end{corollary}

\section{Instantiating the Generator} \label{sec:instantiate}

Having analyzed local encodings and their annihilators, we now move on to constructing hitting set generators using local encodings.
As \cref{cor:hardness to randomness} shows, to construct a hitting set generator for a class $\mathscr{C}$ using a local encoding, it suffices to take an encoding of a circuit that computes a polynomial whose multiples are hard for the target class $\mathscr{C}$.
In this section, we give three constructions of $\VNC^0$-computable generators using local encodings.
One generator hits $\VAC^0$ unconditionally, while the others hit $\VF$ and $\VBP$ under strong but reasonable hardness assumptions.
(In fact, the generators for $\VAC^0$ and $\VF$ are local encodings of the same circuit family.)

Although our work up to this point has proceeded over an arbitrary field $\F$, this section will require $\F$ to have characteristic zero or sufficiently large characteristic.
We omit the large characteristic case for simplicity.
The restriction on the field characteristic stems from the fact that the hitting property for local encodings is inferred from the hardness of multiples of an explicit polynomial $f(\vec{x})$.
Over fields of small positive characteristic, it is an open problem to show that if a polynomial $f(\vec{x})$ is hard to compute, then all multiples of $f$ are similarly hard to compute.
In particular, over a field of characteristic $p > 0$, it is not known if a circuit lower bound for $f(\vec{x})$ implies a comparable lower bound for $f(\vec{x})^p$.
For some limited results in this direction, see \textcite{Andrews20}.

\subsection{Low-Depth Circuits}

In our first application of \cref{cor:hardness to randomness}, we construct a $\VNC^0$-computable generator that hits $\VAC^0$.
To do this, we use the fact that multiples of the determinant are hard for low-depth circuits.
In particular, we make use of the following lower bound.

\begin{theorem}[\cite{LST21, AF22}] \label{thm:vac0 det lb}
    Let $\F$ be a field of characteristic zero.
    Any depth-$\Delta$ circuit that computes a multiple of the $n \times n$ determinant $\det_n(X)$ must be of size at least $n^{(\log n)^{\exp(-O(\Delta))}}$.
\end{theorem}

To construct a $\VNC^0$-computable generator for $\VAC^0$, we first combine \cref{thm:vac0 det lb} with \cref{cor:hardness to randomness} to show that low-depth circuits are hit by local encodings of arithmetic circuits that compute the determinant.

\begin{lemma} \label{lem:local enc hits vac0}
    Let $\F$ be a field of characteristic zero.
    Let $(\Phi_n)_{n \in \naturals}$ be a sequence of arithmetic circuits of size $s(n) \le n^{O(1)}$ such that $\Phi_n$ computes the $n \times n$ determinant $\det_n(X)$.
    Let $(A_n)_{n \in \naturals}$ be a sequence of $n \times n$ matrices.
    Let $\mathcal{G}_n : \F^{n^2 + s(n)} \to \F^{n^2+s(n)+1}$ be the local encoding of $\Phi_n(A_n) = 0$.
    Then $\mathcal{G} = (\mathcal{G}_n)_{n \in \naturals}$ is a $\VNC^0$-computable hitting set generator for $\VAC^0$.
\end{lemma}

\begin{proof}
    The fact that $\mathcal{G}$ is $\VNC^0$-computable follows immediately from \cref{lem:generator parameters}.
    To show that $\mathcal{G}$ hits $\VAC^0$, we will apply \cref{cor:hardness to randomness} and \cref{thm:vac0 det lb}.
    \cref{thm:vac0 det lb} implies that any circuit of depth $\Delta$ computing a multiple of $\det_{n}(X)$ must have size at least
    \[
        n^{(\log n)^{\exp(-O(\Delta))}}.
    \]
    It follows from \cref{cor:hardness to randomness} that $\mathcal{G}_n$ hits circuits of depth $\Delta$ and size 
    \[
        n^{\frac{(\log n)^{\exp(-O(\Delta))}}{\Delta+1}}.
    \]
    This implies that $\mathcal{G}$ hits $\VAC^0$, as we now show.

    Let $(f_n)_{n \in \naturals} \in \VAC^0$.
    By definition, there is a fixed constant $\Delta$ such that $f_n$ can be computed by circuits of depth $\Delta$ and size $n^{O(1)}$.
    Let $m \le n^{O(1)}$ be the number of variables in $f_n$ and let $q$ be the largest integer such that $q^2 + s(q) + 1 \le m$.
    Because $s(q) \le q^{O(1)}$, we have
    \[
        q \ge m^{\Omega(1)} \ge n^{\Omega(1)}.
    \]
    The map $\mathcal{G}_{q}$, padded appropriately as in \cref{rem:padding}, hits $m$-variate circuits of depth $\Delta$ and size
    \[
        q^{\frac{(\log q)^{\exp(-O(\Delta))}}{\Delta+1}} \gg n^{O(1)}.
    \]
    Thus, for sufficiently large $n$, the map $\mathcal{G}_{q}$ hits $f_n$, so $\mathcal{G}$ hits $\VAC^0$ as claimed.
\end{proof}

\cref{lem:local enc hits vac0} gives a $\VNC^0$-computable generator with one bit of stretch that hits $\VAC^0$.
We can improve the stretch of this generator by applying independent copies of the generator in parallel.
This improves the stretch of the generator from $1$ to $n^{1 - \eps}$ for any constant $\eps > 0$.

\begin{theorem} \label{thm:vac0 generator}
    Let $\F$ be a field of characteristic zero and let $\eps > 0$ be a fixed constant.
    There is an explicit sequence of polynomial maps $\mathcal{G} = (\mathcal{G} : \F^{n - n^{1 - \eps}} \to \F^n)_{n \in \naturals}$ such that $\mathcal{G}$ is a $\VNC^0$-computable hitting set generator for $\VAC^0$.
\end{theorem}

\begin{proof}
    Let $\hat{\mathcal{G}} = (\hat{\mathcal{G}}_n : \F^{n-1} \to \F^n)$ be the generator of \cref{lem:local enc hits vac0}, padded appropriately as in \cref{rem:padding}.
    We take $\mathcal{G}_n : \F^{n - n^{1 - \eps}} \to \F^n$ to be the generator obtained by concatenating $n^{1-\eps}$ independent copies of $\hat{\mathcal{G}}_{n^{\eps}}$.
    That is, we set $\mathcal{G}_n$ to be the generator
    \[
        \mathcal{G}_n(\vec{x}, \vec{y}) \coloneqq (\hat{\mathcal{G}}_{n^\eps}(\vec{x}^{(1)}, \vec{y}^{(1)}), \ldots, \hat{\mathcal{G}}_{n^\eps}(\vec{x}^{(n^{1-\eps})},\vec{y}^{(n^{1-\eps})})),
    \]
    where $\vec{x} = \vec{x}^{(1)} \sqcup \cdots \sqcup \vec{x}^{(n^{1 - \eps})}$ is a partition of the $x$ variables and likewise $\vec{y} = \vec{y}^{(1)} \sqcup \cdots \sqcup \vec{y}^{(n^{1-\eps})}$ is a partition of the $y$ variables.
    The $\VNC^0$-computability of $\hat{\mathcal{G}}$ implies that $\mathcal{G}$ is $\VNC^0$-computable.

    The fact that $\mathcal{G}$ hits $\VAC^0$ follows from a straightforward hybrid argument, which we now describe.
    Suppose there is a $p$-family $f = (f_n)_{n \in \naturals} \in \VAC^0$ such that $\mathcal{G}$ fails to hit $f$.
    We will use this to construct a $p$-family $g = (g_n)_{n \in \naturals} \in \VAC^0$ such that $\hat{\mathcal{G}}$ fails to hit $g$, which contradicts \cref{lem:local enc hits vac0}.

    For the sake of notational simplicity, assume without of generality that $f_n$ is a polynomial on $n$ variables $z_1, \ldots, z_n$.
    Let $\vec{z} = \vec{z}^{(1)} \sqcup \cdots \sqcup \vec{z}^{(n^{1 - \eps})}$ be a partition of the $z$ variables.
    Consider the hybrid polynomials
    \begin{align*}
        f_{n,0}(\vec{x}, \vec{y}, \vec{z}) &\coloneqq f_n(\vec{z}^{(1)}, \vec{z}^{(2)}, \ldots, \vec{z}^{(n^{1-\eps})}) \\
        f_{n,1}(\vec{x}, \vec{y}, \vec{z}) &\coloneqq f_n(\hat{\mathcal{G}}_{n^\eps}(\vec{x}^{(1)},\vec{y}^{(1)}), \vec{z}^{(2)}, \ldots, \vec{z}^{(n^{1 - \eps})}) \\
        f_{n,2}(\vec{x}, \vec{y}, \vec{z}) &\coloneqq f_n(\hat{\mathcal{G}}_{n^\eps}(\vec{x}^{(1)},\vec{y}^{(1)}), \hat{\mathcal{G}}_{n^\eps}(\vec{x}^{(2)},\vec{y}^{(2)}), \ldots, \vec{z}^{(n^{1 - \eps})}) \\
        &\vdotswithin{\coloneqq} \\
        f_{n,n^{1 - \eps}}(\vec{x}, \vec{y}, \vec{z}) &\coloneqq f_n(\hat{\mathcal{G}}_{n^\eps}(\vec{x}^{(1)}, \vec{y}^{(1)}), \hat{\mathcal{G}}_{n^\eps}(\vec{x}^{(2)},\vec{y}^{(2)}), \ldots, \hat{\mathcal{G}}_{n^\eps}(\vec{x}^{(n^{1-\eps})},\vec{y}^{(n^{1-\eps})})).
    \end{align*}
    The fact that $\mathcal{G}$ fails to hit $f$ implies that $f_{n,0}(\vec{x}, \vec{y}, \vec{z}) \neq 0$ and $f_{n,n^{1 - \eps}}(\vec{x}, \vec{y}, \vec{z}) = 0$ when $n$ is sufficiently large.
    Thus, there is some $i \in [n^{1-\eps}]$ such that $f_{n,i-1}(\vec{x}, \vec{y}, \vec{z}) \neq 0$ but $f_{n,i}(\vec{x}, \vec{y}, \vec{z}) = 0$.
    Because $\F$ is infinite, there are inputs $\vec{\alpha}^{(1)}, \ldots, \vec{\alpha}^{(i-1)}, \vec{\beta}^{(1)}, \ldots, \vec{\beta}^{(i-1)}$, and $\vec{\gamma}^{(i+1)},\ldots,\vec{\gamma}^{(n^{1-\eps})}$ such that
    \[
        g_n(\vec{z}^{(i)}) \coloneqq f_{n,i-1}(\hat{\mathcal{G}}(\vec{\alpha}^{(1)}, \vec{\beta}^{(1)}), \ldots, \hat{\mathcal{G}}(\vec{\alpha}^{(i-1)}, \vec{\beta}^{(i-1)}), \vec{z}^{(i)}, \vec{\gamma}^{(i+1)}, \ldots, \vec{\gamma}^{(n^{1-\eps})}) \neq 0.
    \]
    Note that $g$ is a projection of $f$, so $g = (g_n)_{n \in \naturals}$ is a $p$-family in $\VAC^0$.
    Furthermore, we have
    \[
        g_n(\hat{\mathcal{G}}_{n^\eps}(\vec{x}^{(i)}, \vec{y}^{(i)})) = f_{i-1}(\hat{\mathcal{G}}(\vec{\alpha}^{(1)}, \vec{\beta}^{(1)}), \ldots, \hat{\mathcal{G}}(\vec{\alpha}^{(i-1)}, \vec{\beta}^{(i-1)}), \hat{\mathcal{G}}(\vec{x}^{(i)}, \vec{y}^{(i)}), \vec{\gamma}^{(i+1)}, \ldots, \vec{\gamma}^{(n^{1-\eps})}) = 0,
    \]
    so $\hat{\mathcal{G}}$ fails to hit $g$, contradicting \cref{lem:local enc hits vac0}.
\end{proof}

\subsection{Formulas}

Next, we construct a $\VNC^0$-computable generator that hits $\VF$.
To do this, we need a polynomial whose multiples all require large formulas.
Formulas are not known to be closed under factorization, so it is not clear if formula lower bounds for an explicit family of polynomials imply lower bounds for multiples of the same family.
For the special case of the determinant, however, lower bounds on border formula complexity imply comparable lower bounds on the (border) formula complexity of multiples, as the following theorem shows.

\begin{theorem}[\cite{AF22}] \label{thm:vf det lb}
    Let $\F$ be a field of characteristic zero.
    Suppose that any border formula which computes the $n \times n$ determinant $\det_n(X)$ must have size at least $n^{\omega(1)}$.
    Then any border formula that computes a multiple of $\det_n(X)$ must have size $n^{\omega(1)}$.
\end{theorem}

Assuming lower bounds on the border formula complexity of the determinant, we can construct a $\VNC^0$-computable generator with one bit of stretch that hits $\VF$.
The proof is essentially the same as in \cref{lem:local enc hits vac0}; the only difference is that the unconditional lower bound of \cref{thm:vac0 det lb} is replaced by the conditional lower bound from \cref{thm:vf det lb}.

\begin{lemma} \label{lem:local enc hits vf}
    Let $\F$ be a field of characteristic zero.
    Assume that any border formula which computes the $n \times n$ determinant $\det_n(X)$ has size $n^{\omega(1)}$.
    Let $(\Phi_n)_{n \in \naturals}$ be a sequence of arithmetic circuits of size $s(n) \le n^{O(1)}$ such that $\Phi_n$ computes the $n \times n$ determinant $\det_n(X)$.
    Let $(A_n)_{n \in \naturals}$ be a sequence of $n \times n$ matrices.
    Let $\mathcal{G}_n : \F^{n^2 + s(n)} \to \F^{n^2+s(n)+1}$ be the local encoding of $\Phi_n(A_n) = 0$.
    Then $\mathcal{G} = (\mathcal{G}_n)_{n \in \naturals}$ is a $\VNC^0$-computable hitting set generator for $\VF$.
\end{lemma}

\begin{proof}
    The fact that $\mathcal{G}_n$ is $\VNC^0$-computable is an immediate consequence of \cref{lem:generator parameters}.
    To show that $\mathcal{G}$ hits $\VF$, we will invoke \cref{cor:hardness to randomness}, \cref{thm:vf det lb}, and the assumed lower bound on the border formula complexity of the determinant.
    Because any border formula that computes the $n \times n$ determinant must have size $n^{\omega(1)}$, \cref{thm:vf det lb} implies that any formula computing a multiple of $\det_{n}(X)$ must have size $t \ge n^{\omega(1)}$.
    It follows from \cref{cor:hardness to randomness} that the map $\mathcal{G}_n$ hits formulas of size $\eps \sqrt{t} \ge n^{\omega(1)}$ for some constant $\eps > 0$.
    In particular, for a $p$-family $(f_n)_{n \in \naturals} \in \VF$, the map $\mathcal{G}_{q(n)}$ (where $q(n) \ge n^{\Omega(1)}$ is the index of the $n$-output generator in $\mathcal{G}$) hits $f_n$ for sufficiently large $n$, since $f_n$ can be computed by a formula of size $n^{O(1)}$.
    Hence $\mathcal{G}$ hits $\VF$ as claimed.
\end{proof}

Once again, we can improve the stretch of the generator of \cref{lem:local enc hits vf} from $1$ to $n^{1 - \eps}$ for any constant $\eps > 0$ by applying $n^{1-\eps}$ copies of the generator in parallel.

\begin{theorem} \label{thm:vf generator}
    Let $\F$ be a field of characteristic zero and let $\eps > 0$ be a fixed constant.
    Assume that any border formula which computes the $n \times n$ determinant $\det_n(X)$ has size $n^{\omega(1)}$. 
    Then there is an explicit sequence of polynomial maps $\mathcal{G} = (\mathcal{G}_n : \F^{n - n^{1 - \eps}} \to \F^n)_{n \in \naturals}$ such that $\mathcal{G}$ is a $\VNC^0$-computable hitting set generator for $\VF$.
\end{theorem}

\begin{proof}
    Let $\hat{\mathcal{G}}$ be the generator of \cref{lem:local enc hits vf}.
    We take $\mathcal{G}_n : \F^{n - n^{1-\eps}} \to \F^n$ to be the concatenation of $n^{1 - \eps}$ independent copies of $\hat{\mathcal{G}}_{n^{\eps}}$.
    The $\VNC^0$-computability of $\hat{\mathcal{G}}$ implies that $\mathcal{G}$ is $\VNC^0$-computable.
    The fact that $\mathcal{G}$ hits $\VF$ follows from the fact that $\hat{\mathcal{G}}$ hits $\VF$, using a hybrid argument in exactly the same manner as in the proof of \cref{thm:vac0 generator}.
\end{proof}

\subsection{Branching Programs}

In our final application, we construct a $\VNC^0$-computable generator that hits $\VBP$.
To do this, we make use of the fact that $\VBP$ is closed under factoring.

\begin{theorem}[\cite{ST21c}] \label{thm:vbp closed under factoring}
    Let $\F$ be a field of characteristic zero.
    Suppose that $f(\vec{x})$ can be computed by an algebraic branching program of size $s$.
    Then every irreducible factor of $f$ can be computed by an algebraic branching program of size $s^{O(1)}$.
\end{theorem}

The contrapositive of \cref{thm:vbp closed under factoring} says that if a polynomial $f$ requires branching programs of super-polynomial size, then the same is true for all multiples of $f$.
This allows us to instantiate \cref{cor:hardness to randomness} under the assumption that there is a $p$-family $(f_n)_{n \in \naturals} \in \VP$ such that $f_n$ requires branching programs of size $n^{\omega(1)}$.

\begin{lemma} \label{lem:local enc hits vbp}
    Let $\F$ be a field of characteristic zero.
    Assume there is a $p$-family $(f_n)_{n \in \naturals} \in \VP$ such that $f_n$ requires branching programs of size $n^{\omega(1)}$.
    Let $(\Phi_n)_{n \in \naturals}$ be a sequence of arithmetic circuits of size $s(n) \le n^{O(1)}$ such that $\Phi_n$ computes $f_n$.
    Let $(\vec{\alpha}^{(n)})_{n \in \naturals}$ be a sequence of points in $\F^n$ and let $(\beta_n)_{n \in \naturals}$ be a sequence of field elements.
    Let $\mathcal{G}_n : \F^{n + s(n)} \to \F^{n+s(n)+1}$ be the local encoding of $\Phi_n(\vec{\alpha}^{(n)}) = \beta_n$.
    Then $\mathcal{G} = (\mathcal{G}_n)_{n \in \naturals}$ is a $\VNC^0$-computable hitting set generator for $\VBP$.
\end{lemma}

\begin{proof}
    The fact that $\mathcal{G}_n$ is $\VNC^0$-computable follows immediately from \cref{lem:generator parameters}.
    To show that $\mathcal{G}_n$ hits $\VBP$, we use \cref{cor:hardness to randomness}, \cref{thm:vbp closed under factoring}, and the assumed lower bound on the $p$-family $(f_n)_{n \in \naturals}$.
    By assumption, any branching program that computes $f_n - \beta_n$ must have size $n^{\omega(1)}$.
    \cref{thm:vbp closed under factoring} implies that any multiple of $f_n - \beta_n$ likewise requires branching programs of size $t \ge n^{\omega(1)}$.
    It follows from \cref{cor:hardness to randomness} that $\mathcal{G}_n$ hits branching programs of size $\eps \sqrt{t} \ge n^{\omega(1)}$ for some constant $\eps > 0$.
    In particular, for a $p$-family $(h_n)_{n \in \naturals} \in \VBP$, the map $\mathcal{G}_q$ (where $q(n) \ge n^{\Omega(1)}$ is the index of the $n$-output generator in $\mathcal{G}$) hits $h_n$ for sufficiently large $n$, since $h_n$ can be computed by a branching program of size $n^{O(1)}$.
    Hence $\mathcal{G}$ hits $\VBP$.
\end{proof}

Once again, we can improve the stretch of the generator of \cref{lem:local enc hits vbp} from $1$ to $n^{1 - \eps}$ by applying $n^{1-\eps}$ copies of the generator in parallel.
The proof is identical to the proofs of \cref{thm:vac0 generator} and \cref{thm:vf generator}, so we omit the details.

\begin{theorem} \label{thm:vbp generator}
    Let $\F$ be a field of characteristic zero and let $\eps > 0$ be a fixed constant.
    Assume there is a $p$-family $(f_n)_{n \in \naturals} \in \VP$ such that $f_n$ requires branching programs of size $n^{\omega(1)}$.
    Then there is an explicit sequence of polynomial maps $\mathcal{G} = (\mathcal{G}_n : \F^{n - n^{1 - \eps}} \to \F^n)_{n \in \naturals}$ such that $\mathcal{G}$ is a $\VNC^0$-computable hitting set generator for $\VBP$.
\end{theorem}

\section{The Ideal Proof System and Cryptographic Generators} \label{sec:geom ips}

We now turn our attention to algebraic proof complexity.
An easy observation, already due to \textcite[Section 4.2]{GKSS17}, shows that a hitting set generator $\mathcal{G}$ for a circuit class $\mathscr{C}$ implies lower bounds against the Geometric $\mathscr{C}$-Ideal Proof System.
If the generator is $\mathscr{D}$-computable, then the hard instance for Geometric $\mathscr{C}$-IPS consists of polynomials that are likewise computable in $\mathscr{D}$, which we term $\mathscr{D}$-equations following \cref{def:d-equation}.
Together with our construction of $\VNC^0$-computable hitting set generators from \cref{sec:instantiate}, this lets us conclude lower bounds against Geometric $\mathscr{C}$-IPS for $\mathscr{C} \in \set{\VAC^0, \VF, \VBP}$ for systems of $\VNC^0$-equations.

\subsection{Geometric IPS Lower Bounds from Hitting Set Generators}

We begin with the relationship between hitting set generators and Geometric IPS, which was already observed by \textcite[Section 4.2]{GKSS17}.
Suppose $\mathcal{G} = (g_1(\vec{y}), \ldots, g_n(\vec{y}))$ is a generator that hits a class of polynomials $\mathscr{C}$.
This means that for every nonzero $f \in \F[\vec{x}]$ satisfying $(f \circ \mathcal{G})(\vec{y}) = 0$, we have the non-membership $f \notin \mathscr{C}$.
Consider the system of polynomial equations $\mathcal{F} = \set{g_1(\vec{y}) = 0, \ldots, g_n(\vec{y}) = 0}$, and suppose that the system $\mathcal{F}$ is unsatisfiable and can be refuted by Geometric IPS.
Let $r(\vec{x})$ be a Geometric IPS refutation of $\mathcal{F}$.
By definition, the polynomial $r$ satisfies the equations 
\begin{align*}
    r(0,\ldots,0) &= 1 \\
    r(g_1(\vec{y}), \ldots, g_n(\vec{y})) &= 0.
\end{align*}
The polynomials that satisfy these equations are precisely the annihilators of $\mathcal{G}$ whose constant term is $1$.
Because the refutation $r$ is a nonzero annihilator of $\mathcal{G}$ and $\mathcal{G}$ hits $\mathscr{C}$, we conclude that $r \notin \mathscr{C}$, which is a lower bound on the complexity of $r$.

Thus, Geometric IPS refutations are simply annihilators with a nonzero constant term.
It seems reasonable to conjecture that if there is an easily-computable annihilator, and an annihilator with a nonzero constant term exists, then there is also an easily-computable annihilator with a nonzero constant term.
Put more informally, it is not outlandish to think that if we can prove Geometric IPS lower bounds, then we can upgrade the hard system of equations into a hitting set generator.
We are unable to prove such a statement.
Despite this, we believe that studying Geometric IPS is a meaningful stepping stone towards obtaining hitting set generators.
This very paper is evidence for our belief: we first proved lower bounds against Geometric IPS, and only later realized the same construction yields hitting set generators with a minor variation on the proof.

We formalize the preceding argument in the following lemma.

\begin{lemma}[\cite{GKSS17}] \label{lem:geom ips from hsg}
    Let $\F$ be a field and let $\mathscr{C}$ and $\mathscr{D}$ be complexity classes.
    Let $\mathcal{G} = (\mathcal{G}_n : \F^{\ell(n)} \to \F^{n})_{n \in \naturals}$ be a $\mathscr{D}$-computable hitting set generator for $\mathscr{C}$.
    Then the family of equations $\mathcal{F} = (\mathcal{F}_n)_{n \in \naturals}$ given by
    \[
        \mathcal{F}_n = \set{g_{n,1}(\vec{y}) = 0, \ldots, g_{n,n}(\vec{y}) = 0},
    \]
    where $\mathcal{G}_n = (g_{n,1}, \ldots, g_{n,n})$, is a family of $\mathscr{D}$-computable equations that is hard for Geometric $\mathscr{C}$-IPS.
\end{lemma}

\begin{proof}
    Let $(r_n)_{n \in \naturals}$ be a $p$-family such that $r_n(\vec{x})$ is a Geometric IPS refutation of $\mathcal{F}_n$.
    The definition of a Geometric IPS refutation implies
    \begin{align*}
        r_n(0,\ldots,0) &= 1 \\
        r_n(g_{n,1}(\vec{y}), \ldots, g_{n,n}(\vec{y})) &= 0.
    \end{align*}
    The first equality implies $r_n(\vec{x}) \neq 0$, while the second implies $r_n \in \Ann{\mathcal{G}_n}$.
    Because $\mathcal{G}$ hits $\mathscr{C}$, it follows that $r \notin \mathscr{C}$, hence there is no Geometric $\mathscr{C}$-IPS refutation of $\mathcal{F}$.
\end{proof}

We note that in the setting of \cref{lem:geom ips from hsg}, neither the explicitness nor the stretch of the hitting set generator are of primary concern.
The resulting lower bounds are interesting even when the generator is non-explicit and has stretch $1$.

\subsection{Lower Bounds for Geometric IPS}

We now apply \cref{lem:geom ips from hsg} to the hitting set generators of \cref{sec:instantiate}, obtaining systems of $\VNC^0$-equations that are hard for subsystems of Geometric IPS.
The lower bounds we obtain here directly correspond to the generator constructions of \cref{sec:instantiate}: we obtain an unconditional lower bound against Geometric $\VAC^0$-IPS and conditional lower bounds against Geometric $\VF$-IPS and Geometric $\VBP$-IPS.

We stress that because Geometric IPS is an incomplete proof system, it is not enough for our hard instances to be unsatisfiable; we must ensure that Geometric IPS is powerful enough to refute these systems of equations in the first place.
For the systems we consider, this property is easy to verify, as the polynomial that generates the annihilator ideal $\Ann{\mathcal{G}}$ is a valid Geometric IPS refutation.

We start with a system of $\VNC^0$-computable equations that are unconditionally hard to refute in Geometric $\VAC^0$-IPS.

\begin{theorem} \label{thm:vac0 geom ips lb}
    Let $\F$ be a field of characteristic zero.
    Let $(\Phi_n)_{n \in \naturals}$ be a sequence of arithmetic circuits of size $s(n) \le n^{O(1)}$ such that $\Phi_n$ computes the $n \times n$ determinant $\det_n(X)$.
    Let $(A_n)_{n \in \naturals}$ be a sequence of $n \times n$ matrices such that $\det(A_n) \neq 0$ for all $n \in \naturals$.
    Let $\mathcal{G}_n : \F^{n^2 + s(n)} \to \F^{n^2+s(n)+1}$ be the local encoding of $\Phi_n(A_n) = 0$.
    Let $\mathcal{F}_n$ be the system of equations given by
    \[
        \mathcal{F}_n = \set{g_{n,1}(\vec{y}) = 0, \ldots, g_{n, n^2 + s(n) + 1}(\vec{y}) = 0},
    \]
    where $\mathcal{G}_n = (g_{n,1}, \ldots, g_{n, n^2 + s(n) + 1})$.
    Then $\mathcal{F} = (\mathcal{F}_n)_{n \in \naturals}$ is a sequence of unsatisfiable $\VNC^0$-equations that can be refuted by Geometric $\VP$-IPS, but cannot be refuted by Geometric $\VAC^0$-IPS.
\end{theorem}

\begin{proof}
    The fact that $\mathcal{F}_n$ is a system of $\VNC^0$-equations that cannot be refuted by Geometric $\VAC^0$-IPS follows from \cref{lem:local enc hits vac0,lem:geom ips from hsg}.
    To show that $\mathcal{F}_n$ can be refuted by Geometric IPS, let $h_n(\vec{z}) \in \Ann{\mathcal{G}_n}$ be the generator of the ideal $\Ann{\mathcal{G}_n}$.
    By \cref{prop:ann ideal generator}, we know that $h_n$ is given by
    \[
        h_n(\vec{z}) = z_{n^2+s+1} - \det(Z_{1 \ldots n^2} + A_n) + g(\vec{z}),
    \]
    where $Z_{1 \ldots n^2}$ denotes an $n \times n$ matrix formed from the variables $z_1, \ldots, z_{n^2}$ and $g(\vec{z})$ is a polynomial that satisfies $g(\vec{0}) = 0$.
    The constant term of $h_n$ is given by
    \[
        h_n(\vec{0}) = -\det(A_n) \neq 0.
    \]
    In particular, the polynomial $\frac{-1}{\det(A_n)}h_n(\vec{z})$ is a Geometric IPS refutation of $\mathcal{F}_{n}$.
    The soundness of Geometric IPS implies that $\mathcal{F}_n$ is unsatisfiable.
    \cref{prop:ann ideal generator} implies that $\frac{-1}{\det(A_n)}h_n(\vec{z})$ can be computed by a circuit of size $O(s) \le n^{O(1)}$, so $(h_n)_{n \in \naturals} \in \VP$.
    Hence Geometric $\VP$-IPS can refute $\mathcal{F}$ as claimed.
\end{proof}

Next, we prove a conditional lower bound against Geometric $\VF$-IPS.
Our lower bound for Geometric $\VF$-IPS holds assuming a super-polynomial lower bound on the border formula complexity of the determinant.
This is the same condition used in \cref{thm:vf generator} to design a $\VNC^0$-computable hitting set generator against $\VF$.

\begin{theorem} \label{thm:vf geom ips lb}
    Let $\F$ be a field of characteristic zero.
    Assume that any border formula which computes the $n \times n$ determinant $\det_n(X)$ must have size $n^{\omega(1)}$.
    Let $(\Phi_n)_{n \in \naturals}$ be a sequence of arithmetic circuits of size $s(n) \le n^{O(1)}$ such that $\Phi_n$ computes the $n \times n$ determinant $\det_n(X)$.
    Let $(A_n)_{n \in \naturals}$ be a sequence of $n \times n$ matrices such that $\det(A_n) \neq 0$ for all $n \in \naturals$.
    Let $\mathcal{G}_n : \F^{n^2 + s(n)} \to \F^{n^2+s(n)+1}$ be the local encoding of $\Phi_n(A_n) = 0$.
    Let $\mathcal{F}_n$ be the system of equations given by
    \[
        \mathcal{F}_n = \set{g_{n,1}(\vec{y}) = 0, \ldots, g_{n, n^2 + s(n) + 1}(\vec{y}) = 0},
    \]
    where $\mathcal{G}_n = (g_{n,1}, \ldots, g_{n, n^2 + s(n) + 1})$.
    Then $\mathcal{F} = (\mathcal{F}_n)_{n \in \naturals}$ is a sequence of unsatisfiable $\VNC^0$-equations that can be refuted by Geometric $\VP$-IPS, but cannot be refuted by Geometric $\VF$-IPS.
\end{theorem}

\begin{proof}
    The fact that $\mathcal{F}_n$ is a system of $\VNC^0$-equations that cannot be refuted by Geometric $\VF$-IPS follows from the assumed lower bound on the border formula complexity of the determinant together with \cref{lem:local enc hits vf,lem:geom ips from hsg}.
    This is the same system of equations that appears in \cref{thm:vac0 geom ips lb}, so the fact that $\mathcal{F}$ can be refuted by Geometric $\VP$-IPS follows from \cref{thm:vac0 geom ips lb}.
\end{proof}

Our final application to Geometric IPS lower bounds is a conditional lower bound against Geometric $\VBP$-IPS.
We use the same assumption as in \cref{thm:vbp generator}, where we constructed a $\VNC^0$-computable generator that hits $\VBP$.
That is, we assume there is a $p$-family $(f_n)_{n \in \naturals} \in \VP$ such that the branching program complexity of $f_n$ grows super-polynomially in $n$.

\begin{theorem} \label{thm:vbp geom ips lb}
    Let $\F$ be a field of characteristic zero.
    Assume there is a $p$-family $(f_n)_{n \in \naturals} \in \VP$ such that $f_n$ requires branching programs of size $n^{\omega(1)}$.
    Let $(\Phi_n)_{n \in \naturals}$ be a sequence of arithmetic circuits of size $s(n) \le n^{O(1)}$ such that $\Phi_n$ computes $f_n$.
    Let $(\vec{\alpha}^{(n)})_{n \in \naturals}$ be a sequence of points in $\F^n$ and $(\beta_n)_{n \in \naturals}$ be a sequence of field elements such that $f_n(\vec{\alpha}^{(n)}) \neq \beta_n$ for all $n$.
    Let $\mathcal{G}_n : \F^{n + s(n)} \to \F^{n+s(n)+1}$ be the local encoding of $\Phi_n(\vec{\alpha}^{(n)}) = \beta_n$.
    Let $\mathcal{F}_n$ be the system of equations given by
    \[
        \mathcal{F}_n = \set{g_{n,1}(\vec{y}) = 0, \ldots, g_{n, n + s(n) + 1}(\vec{y}) = 0},
    \]
    where $\mathcal{G}_n = (g_{n,1}, \ldots, g_{n, n + s(n) + 1})$.
    Then $\mathcal{F} = (\mathcal{F}_n)_{n \in \naturals}$ is a sequence of unsatisfiable $\VNC^0$-equations that can be refuted by Geometric $\VP$-IPS, but cannot be refuted by Geometric $\VBP$-IPS.
\end{theorem}

\begin{proof}
    The fact that $\mathcal{F}_n$ is a system of $\VNC^0$-equations that cannot be refuted by Geometric $\VBP$-IPS follows from \cref{lem:local enc hits vbp,lem:geom ips from hsg}.
    To show that $\mathcal{F}_n$ can be refuted by Geometric IPS, let $h_n(\vec{z}) \in \Ann{\mathcal{G}_n}$ be the generator of $\Ann{\mathcal{G}_n}$.
    By \cref{prop:ann ideal generator}, we know that $h_n$ is given by
    \[
        h(\vec{z}) = z_{n+s+1} - f_n(z_1 + \alpha_1, \ldots, z_n + \alpha_n) + g(\vec{z}) + \beta_n,
    \]
    where $g(\vec{z})$ is a polynomial satisfying $g(\vec{0}) = 0$.
    The constant term of $h_n$ is given by
    \[
        h_n(\vec{0}) = \beta_n-f_n(\vec{\alpha}) \neq 0,
    \]
    so $\frac{1}{\beta_n - f_n(\vec{\alpha})} h_n(\vec{0}) = 1$.
    In particular, the polynomial $\frac{1}{\beta_n-f_n(\vec{\alpha})} h_n(\vec{z})$ is a Geometric IPS refutation of $\mathcal{F}_{n}$.
    The soundness of Geometric IPS implies that $\mathcal{F}_n$ is unsatisfiable.
    Moreover, \cref{prop:ann ideal generator} implies that $\frac{1}{\beta_n-f_n(\vec{\alpha})} h_n$ can be computed by a circuit of size $O(s) \le n^{O(1)}$, so $(h_n)_{n \in \naturals} \in \VP$. Hence Geometric $\VP$-IPS can refute $\mathcal{F}_n$.
\end{proof}

\section{Conclusion and Open Problems} \label{sec:conclusion}

In this work, we gave new constructions of $\VNC^0$-computable hitting set generators for fairly strong circuit classes, whose proof of correctness is based on lower bounds for circuit complexity rather than degree.
Our constructions followed a general paradigm: from a separation of complexity classes $\mathscr{C} \subsetneq \mathscr{D} \subseteq \VP$, we constructed a $\VNC^0$-computable generator that hits $\mathscr{C}$ by using a local encoding of a $p$-family in the difference $\mathscr{D} \setminus \mathscr{C}$.
The seed length of the generator corresponds to the circuit complexity of the hard $p$-family in $\mathscr{D} \setminus \mathscr{C}$.
This means that we cannot hope to use local encodings to construct generators that hit $\VP$.
Because $\VP$ is a class of low-degree polynomials, the same argument as in \cref{subsec:annihilator degree bounds} shows that degree lower bounds for annihilators can be used to construct $\VNC^0$-computable generators for $\VP$.
Can one construct $\VNC^0$-computable generators for $\VP$ whose correctness is based on circuit complexity, not on degree bounds?

A natural approach to this problem would be to take a known hitting set generator $\mathcal{G}$ for $\VP$ and compile it into a different generator $\mathcal{G}'$ that may have worse stretch than $\mathcal{G}$, but improves on the complexity of $\mathcal{G}$ and retains its ability to hit $\VP$.
Such a compiler underlies the $\NC^0$-computable one-way functions and pseudorandom generators constructed by \textcite{AIK06}.
We do not know how to construct such a compiler in the algebraic setting.
One can obtain a new generator $\mathcal{G}'$ by encoding the generator $\mathcal{G}$ by extension variables, similar to how we represent circuits by local encodings, and the resulting $\mathcal{G}'$ is indeed $\VNC^0$-computable.
However, it's not obvious how to show that lower bounds on the annihilators of $\mathcal{G}$ imply lower bounds on annihilators of the new generator $\mathcal{G}'$.
It may be the case that by transforming $\mathcal{G}$ to $\mathcal{G}'$, we introduce new annihilating polynomials of lower complexity.
Our work shows that this is not possible for generators of the form $\vec{x} \mapsto (\vec{x}, f(\vec{x}))$, but it is less clear how to reason about other base generators $\mathcal{G}$.

For example, if we take $\mathcal{G}$ to be the Kabanets--Impagliazzo generator instantiated with the permanent, can we show that its local encoding $\mathcal{G}'$ hits $\VP$, assuming circuit lower bounds for the permanent?
The obvious approach to proving this would be to take an annihilator of the local encoding $\mathcal{G}'$ and obtain from it an annihilator of the base generator $\mathcal{G}$ by substituting a gate variable $y_i$ with the polynomial $f_i(\vec{x})$ computed at the $i$\ts{th} gate, just as we did in the proof of \cref{prop:ann ideal generator}.
If the annihilator of $\mathcal{G}'$ remains nonzero under this substitution, we can conclude a lower bound on its complexity, but it may be the case that the annihilator becomes zero under this substitution, in which case we do not know how to conclude a lower bound on its complexity.

Another natural question in this line of work is to improve on the parameters of our $\VNC^0$-computable generators.
We are only able to construct generators with sublinear stretch.
Can this be improved to linear or superlinear stretch? 
Or are there inherent tradeoffs between the stretch, degree, and complexity of a hitting set generator?

In our application to lower bounds for the Ideal Proof System, we were only able to establish lower bounds against the Geometric Ideal Proof System.
Can these lower bounds be strengthened to hold for subsystems of the Ideal Proof System, such as $\VAC^0$-IPS?
Although polynomial identity testing has a clear and simple relationship to Geometric IPS, it is not obvious if there is a useful application of hitting set generators towards proving lower bounds for the Ideal Proof System.

Finally, constructions of $\NC^0$-computable one-way functions and pseudorandom generators have found numerous applications throughout complexity theory and cryptography.
Are there useful applications of $\VNC^0$-computable generators within algebraic complexity, or in complexity theory more broadly?

\subsection*{Acknowledgments}

We thank Ramprasad Saptharishi for a simpler formulation and proof of \cref{lem:annihilator inductive step}, Rafael Oliveira for helpful conversations on Nullstellensatz degree bounds, and Abhibhav Garg for useful feedback that improved the presentation of this work.

\printbibliography

\end{document}